\documentclass[aps,prl,reprint,superscriptaddress,nofootinbib]{revtex4-2}

\usepackage[utf8]{inputenc}
\usepackage[english]{babel}
\usepackage{mathtools}
\usepackage{physics}
\usepackage{float}
\usepackage{placeins}
\usepackage{braket}
\usepackage{bbm}
\usepackage{bm}
\usepackage{amsbsy}
\usepackage{amsthm}
\usepackage{amssymb}
\usepackage{amsfonts}
\usepackage{amsmath}
\usepackage{dsfont} 
\usepackage{graphicx} 
\usepackage{hhline}
\usepackage{epsfig}
\usepackage{epstopdf}
\usepackage{dsfont}
\usepackage{multibib}
\usepackage{xcolor}
\usepackage{color}
\usepackage{verbatim}
\usepackage{nomencl}
\usepackage[colorlinks]{hyperref}
\usepackage{float}
\usepackage[normalem]{ulem}
\usepackage{multirow}
\usepackage{makecell}
\usepackage{soul}

\newtheorem{theorem}{Theorem}

\newtheorem{lemma}{Lemma}

\begin{document}

\title{Kinetically constrained superradiance}
\author{Luis Fernando dos Prazeres}
\email[]{ldospraz@buffalo.edu}
\affiliation{Department of Physics, The State University of New York at Buffalo, Buffalo, New York 14260, USA}
\author{Hossein Hosseinabadi}
 \affiliation{Institut f{\"u}r Physik, Johannes Gutenberg-Universit{\"a}t Mainz, 55099 Mainz, Germany}
\author{Jamir Marino}
\affiliation{Department of Physics, The State University of New York at Buffalo, Buffalo, New York 14260, USA}

\begin{abstract}
We introduce kinetically constrained superradiance, a form of cooperative emission in which interactions imprint configuration-dependent energy shifts on optical transitions, splitting Dicke superradiance into multiple, frequency-resolved collective decay channels. Each channel selectively radiates from distinct many-body spin configurations, generating a hierarchy of dissipative time scales and   sequential relaxation dynamics. Unlike conventional superradiance, where permutation symmetry enforces relaxation to a trivial steady state, configuration-selective emission can trap finite-momentum spin-wave excitations and stabilize long-lived entanglement. Remarkably, these correlations are generated purely by dissipation  in the absence of entangling coherent dynamics. Our results point to modern superradiant experiments as scalable resources for dissipative engineering of correlated quantum states.
\end{abstract}

\maketitle

\textit{\textbf{Introduction}--}
Harnessing dissipation as a resource to prepare quantum states with long-lived coherence and spatially extended   entanglement is a long-sought goal at the interface of ultra-cold atoms, solid state quantum information, and quantum computing~\cite{harrington2022engineered}. Many-body dissipative state engineering aims at designing open-system dynamics whose long-lived transients or steady states encode  quantum correlations generated   by dissipative channels rather than   by many-body interactions, therefore converting the ubiquitous adversary of any quantum information application into an ally~\cite{verstraete2009quantum,diehl2008quantum}. 
Despite intense theoretical activity and a rapidly expanding toolbox of reservoir-engineering strategies~\cite{lin2013dissipative,stannigel2012driven,budich2015dissipative,reiter2016scalable,plankensteiner2015selective,marr2003entangled,pocklington2025accelerating,kastoryano2011dissipative}, scalable experimental implementations remain challenging, with only a few notable exceptions~\cite{mi2024stable,ma2019dissipatively}. In this work, we build a   framework connecting superradiance and many-body dissipative state preparation, with the  goal of leveraging modern experimental advances in the former to develop new strategies for quantum state engineering in open systems.

In Dicke’s original proposal~\cite{Dicke1954,gross1982superradiance} an ensemble of atoms coupled to a common radiation field emits synchronously, producing an emission burst with  intensity  scaling quadratically with the number of atoms, a hallmark of cooperative behavior~\cite{gross1982superradiance,Dicke1954}. Dissipation acts through a single collective channel that preserves permutation symmetry among emitters and drives relaxation into a featureless asymptotic state. In recent years, this phenomenon has been experimentally demonstrated across a broad range of platforms, including atomic ensembles in free space~\cite{DeVoe_SR1996,Guerin_SR2016,ferioli2023non, kim2025super_subradiant},  emitters coupled to optical waveguides~\cite{Sheremet_WaveGuide2023,Goban_WaveGuideSR2015,Solano_WaveGuideSR2016,Tiranov_WaveGuideSR2023,liedl2024observation}, and solid-state systems such as color centers in diamond~\cite{angerer2018superradiant,Pallmann_NVSR2024,kersten2026self}. The established universality \cite{Kumlin2025,Li2025,Malz2022,Masson_Universality2022,Lambert2016,Clopez2023,Lohof2023,Windt2025,RBigorda2022,Kirton2017,Andolina2024,Agarwal_directional2024,Sinha_nonMarkov2020}, together with such trans-disciplinary degree of experimental control now attainable, suggests that superradiance has reached a stage of maturity conducive to  applications.

\begin{figure}[t]
    \centering
    \includegraphics[width=0.99\linewidth]{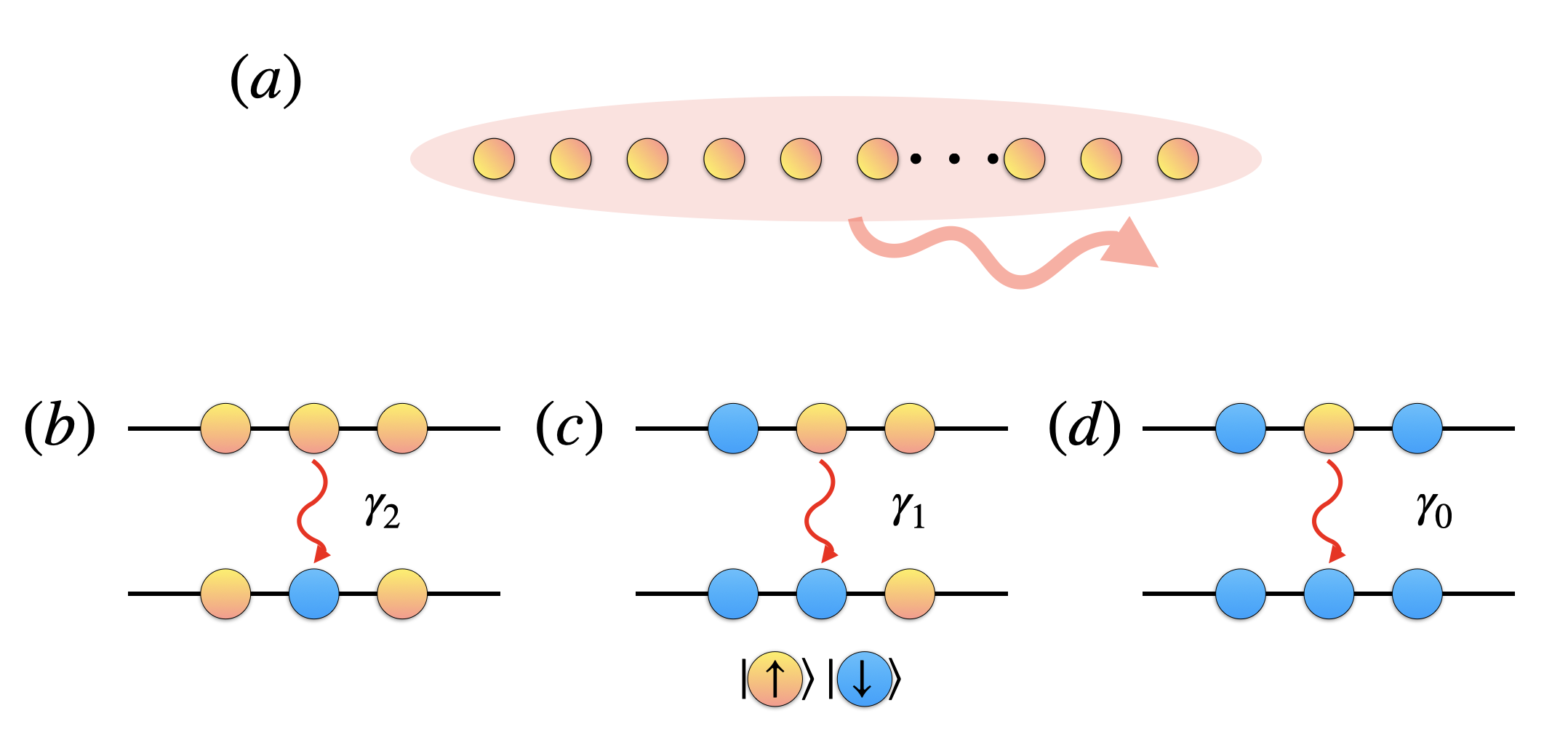}
    \caption{ {Panel $(a)$   depicts the cooperative radiation considered in this work. }Panels $(b)–(d)$ illustrate the three conditional spin-flip processes contributing to such collective emission (with associated decay rates $\gamma_\xi$) selected by the local projectors $P_j^\xi$ for $\xi=2,1,0$; these correspond  to emission events conditioned on the number of excited nearest neighbors, $\xi$.}
    \label{fig:schem}
\end{figure}

Following this route, we consider here a minimal model of an ensemble of two-level atoms coupled by short-range interactions that imprint configuration-dependent energy shifts on optical transitions, such that the frequency of an emitted photon encodes information on the atomic excitations   in the emitter’s surroundings~\cite{Nill22}. As a result, conventional Dicke superradiance breaks down into multiple  emission channels, each of them bright or dark to distinct many-body spin configurations, that can therefore selectively participate to relaxational dynamics. We term this phenomenon  \emph{kinetically constrained} (KC) \emph{superradiance},  and we show that it can drive the system into dissipative steady states that can trap finite-momentum, spin-wave–like collective excitations, and sustain  entanglement over long times, in a   setting where interactions alone could not build   quantum correlations.
 \\

\textit{\textbf{Model}--} {We study  cooperative emission from a one-dimensional array of $N$
two-level atoms with nearest-neighbor interactions, $J$,  and periodic boundary conditions (PBC)}.
The atomic Hamiltonian reads
\begin{equation}
H_A = \Delta \sum_{j=1}^{N} n_j
+ J \sum_{j=1}^N n_j n_{j+1},
\label{eq:ham}
\end{equation}
where $n_j=(1+\sigma_j^z)/2$ is the excitation number operator at site
$j$, $\sigma_j^z|\uparrow(\downarrow)\rangle_j=\pm|\uparrow(\downarrow)\rangle_j$ {and $\Delta, \ J > 0$}.
We  assume that   all atoms couple collectively to the radiation field \cite{Dicke1954} [cf. also Fig \ref{fig:schem} $(a)$]. {A realistic implementation of our model would require to solve a kinetically constrained variant of cavity-QED, as we discuss in the final section of this Letter. This would largely complicate the problem adding further couplings and processes that would obscure the physical picture emerging from   our minimal modeling. This promising direction is left for future work, while here we aim at a proof of principle.}

The presence of atomic interactions   can modify cooperative emission
by making it a state dependent process.
Because interactions shift the energy released upon decay, the
frequency of the emitted photon depends on the number of excited
neighbors of the decaying atom.
In one dimension, an atom can therefore decay through three distinct
processes, emitting photons at frequencies
$\omega_\xi = \Delta + \xi J$, with $\xi = 2,1,0$,
where $\xi$ denotes the number of excited nearest neighbors (cf. panels (b-d) of Fig.~\ref{fig:schem}).
{In contrast to Ref.~\cite{Nill22} where this form of dissipation  is treated as a local process, in our work we consider the case in which that all atoms couple to a common radiation mode, promoting  jump operators to collective form. This apparently straightforward extension is sufficient to enable a form of super-radiance that can support entanglement and trap spin excitations within dark manifolds. A phenomenology distinct from traditional forms of super-radiance~\cite{Li2025,Malz2022,Masson_Universality2022, gross1982superradiance}, and not present in the local dissipative version of this model~\cite{Nill22}.}

{By performing a rotating wave approximation  valid when $J\gg\gamma_\xi$  (see~\cite{SM} for further details),} we arrive at the 
  Lindblad master-equation  \cite{breuer2002,stefanini2025lindblad,Nill22} 
\begin{equation}\label{eq:Lind}
   \dot \rho = -i[H_A, \rho] + \sum_{\xi} \gamma_\xi\left( S_\xi^-\rho S^+_\xi - \frac{1}{2}\left\{S^+_\xi S^-_\xi, \rho \right\}\right),
\end{equation}
with $S^-_\xi = \sum_{j=1}^N P_j^\xi\,\sigma_j^- \ \text{and} \ \xi = 2,1,0$. 
State-dependent   decay processes are implemented by modifying the spontaneous-emission jump operators $\sigma_j^-$ with local projectors $P_j^\xi$, which restrict emission at site $j$ to configurations with exactly $\xi$ excited neighbors. These projectors select the local environment prior to emission, ensuring that the frequency of the emitted photon matches $\omega_\xi$. Specifically, the collective jump operators are constructed from
locally constrained emission processes. The contribution associated with $\xi=2$ selects sites $j$ whose two neighboring sites are excited [Fig.~\ref{fig:schem} $(b)$], implemented through the projector $P^2_j = n_{j-1} n_{j+1}$; the $\xi=1$ channel selects configurations with exactly one excited neighbor [Fig.~\ref{fig:schem}$(c)$], with $P^1_j = n_{j-1}(1-n_{j+1}) + (1-n_{j-1})n_{j+1}$; and the $\xi=0$ channel selects sites surrounded by no excitations [Fig.~\ref{fig:schem} $(d)$], with $P^0_j = (1-n_{j-1})(1-  n_{j+1})$. Each decay channel is associated with a rate
$\gamma_\xi = \Gamma (\Delta + \xi J)^3$ ({see ~\cite{SM} for details.}), where the prefactor $\Gamma$ contains the microscopic details of the light–matter coupling (see~\cite{SM}). 
In the weak interaction regime $J \ll \Delta$, the interaction-induced shifts are small and the decay rates are nearly identical, $\gamma_\xi \simeq \Gamma \Delta^3$.
For stronger interactions $J \gtrsim \Delta$, however, the decay becomes strongly state-dependent, leading to a   hierarchy $\gamma_2 \gg \gamma_1 \gg \gamma_0$, and a corresponding separation of dissipative time scales, whose dynamical consequences are discussed below. {Here   free-space emission has been assumed for simplicity in order to convey the proof-of-principle of the main mechanisms at the core of our work. In a more realistic  cavity-QED implementation (cf. dedicated section) the separation of timescales is achieved following a different route, but the overall physical picture would remain unaltered.}
 \\

\begin{figure}[t]
    \centering
    \includegraphics[width=\linewidth]{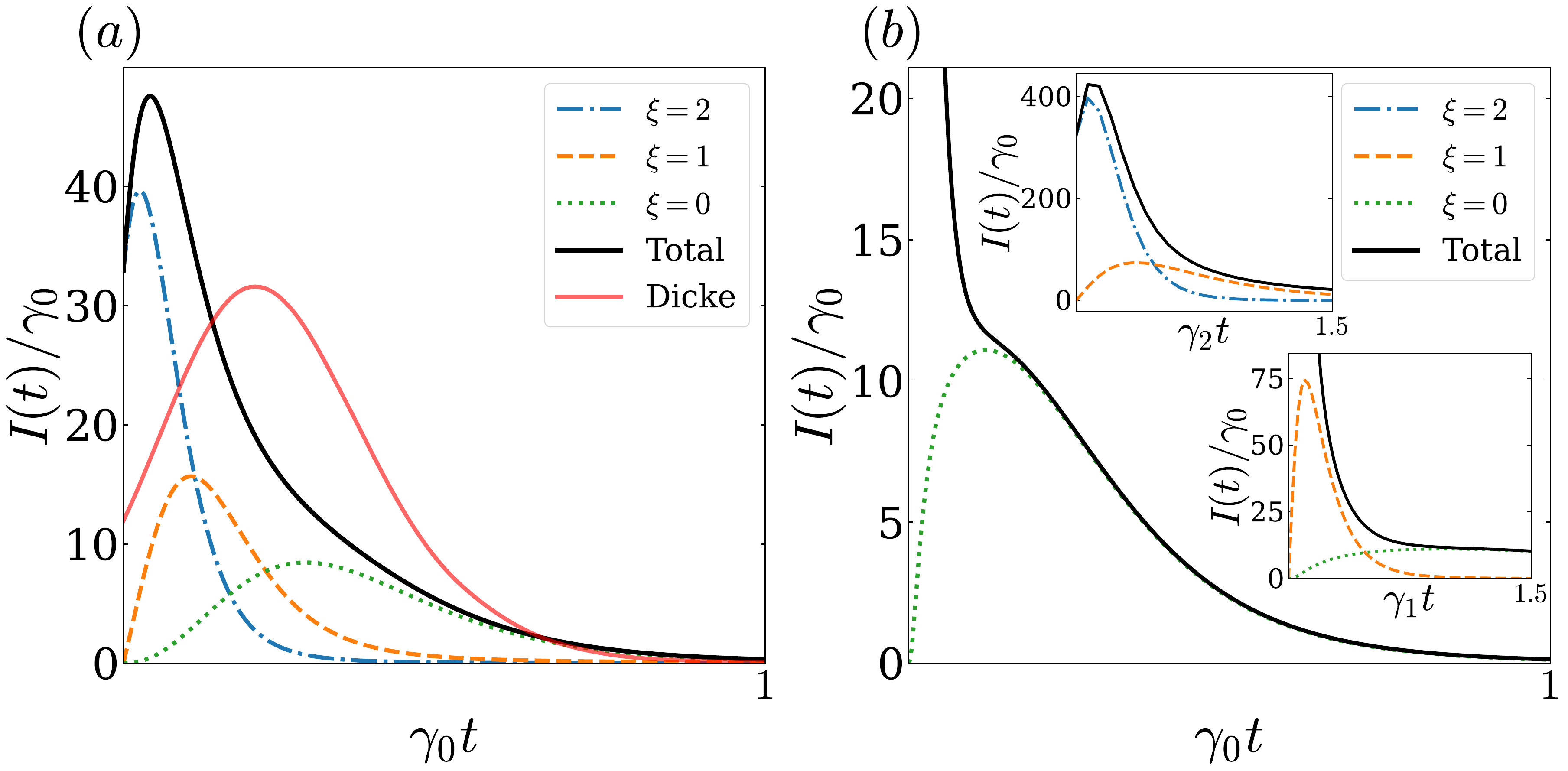}
\caption{
Superradiant bursts.
Left panel (a): Emission intensity for $J/\Delta=0.2$, showing the frequency-resolved channels $\xi=0,1,2$ and the total intensity, together with Dicke superradiance for reference, {whose decay rate is set equal to $\gamma_0$ to allow for a closer comparison with the dynamics of KC super-radiance}.  
Right panel (b): Strongly interacting regime ($J/\Delta=1$) in different time units. The main panel is plotted on the slowest time units ($ \gamma_0^{-1}$).
\textbf{Insets}: Intensities in two different time units,  $ \gamma_2^{-1}$ (upper inset), and $ \gamma_1^{-1}$ (bottom inset), highlighting the hierarchy $\gamma_2 \gg \gamma_1 \gg \gamma_0$ and the metastable plateau (bottom inset) before the decay (main panel). {Dynamics are simulated in exact numerics. Numerical parameters: $N =12$;  $\Gamma = 0.1$.}}

    \label{fig:bursts}
\end{figure}

\FloatBarrier

\begin{figure}[t]
    \centering
    \includegraphics[width=\linewidth]{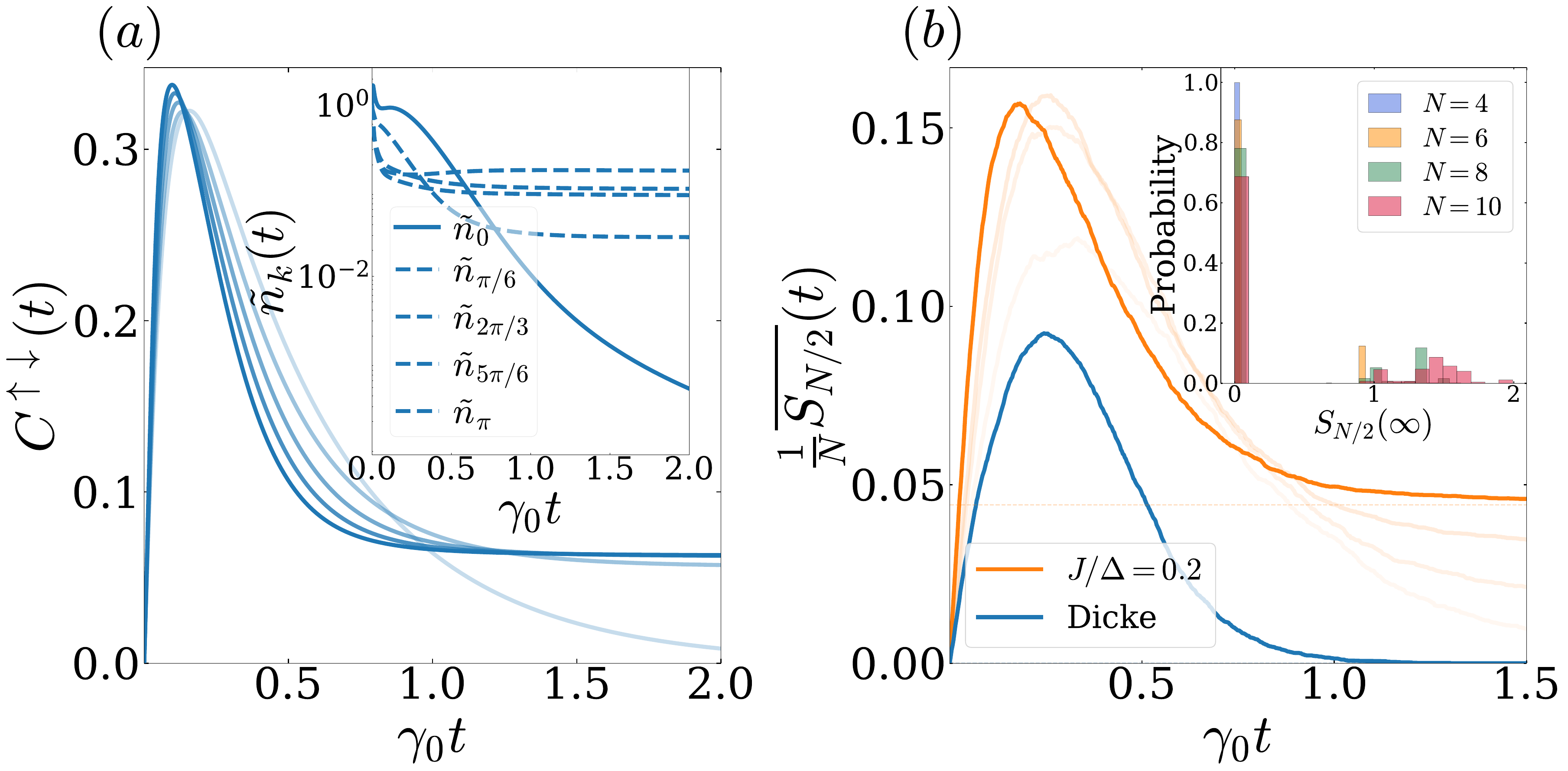}
    \caption{{
(a) Excitation dynamics. Non-vanishing antiferromagnetic correlations in the dark steady-state manifold serve as an entanglement witness.
\textbf{Inset:} Dynamics of $\tilde{n}_k(t)$ for different $k$-modes in a chain of $N=12$ spins with $J/\Delta = 0.2$. Curves with increasing $N=4,6,8,10 ,12$ correspond to increasing opacity of blue. Dynamics are extracted with exact diagonalization of the Liouvillian.
(b) Trajectory-level entanglement entropy. Time evolution of the half-system entanglement entropy for $N=10$, averaged over individual quantum-jump trajectories in the weakly interacting regime ($J/\Delta = 0.2$). Curves with increasing $N=4,6,8,10$ corresponds to increasing opacity.
\textbf{Inset:} Steady-state entanglement entropy distribution from $10^3$ trajectories for different system sizes. In all simulations we have set $\Gamma = 0.1$.}
}
    \label{fig:QP_ent}
\end{figure}

\textit{\textbf{Kinetically constrained superradiance}--} The collective decay generated by the constrained jump operators in
Eq.~\eqref{eq:Lind} produces a modified form of superradiance,
characterized by distinct kinetically constrained (KC) bursts and a
hierarchy of dissipative time scales.
We characterize this behavior through the frequency-resolved emission
rate $I_\xi(t)=\gamma_\xi\langle S^+_\xi S^-_\xi\rangle$, which quantifies the rate
at which photons of frequency $\omega_\xi$ are emitted.
This quantity is directly related to excitation number dynamics
$n=\sum_{j=1}^N n_j$, whose time derivative obeys
$\dot n = - \sum_{\xi  =0}^2  I_\xi(t) \equiv -I(t)$.
Starting for an initially fully inverted state  $\ket{\uparrow}^{\otimes N}$, we show the intensity of light
emitted at each frequency $\omega_\xi$ (dashed curves in
Fig.~\ref{fig:bursts} both panels), together with the total emission rate
$I(t)$ (solid black curves).
As shown in both panels of Fig.~\ref{fig:bursts}, the constrained jump operators give
rise to a hierarchy of emission channels, with superradiant peaks
ordered by decreasing frequency $\omega_\xi$ ($\xi=2,1,0$).
This ordering reflects the structure of KC emission: starting from a
fully inverted state, the early-time dynamics is dominated by decay
through $S^-_2$, while emission through $S^-_1$ and, subsequently,
$S^-_0$ emerges as the system dynamically accesses the corresponding
bright subspaces.

In both the weak and strong interaction regimes (cf. Fig.~\ref{fig:bursts}),   emission intensity  exhibits the familiar superradiant scaling
$I_{\mathrm{peak}}\propto N^2$, cf. \cite{SM} and Refs.~\cite{gross1982superradiance, Masson20, Masson24}.
This scaling reflect the superradiant response of a fully inverted
state, in which collective decay dynamically builds phase correlations across the atomic ensemble.
The subsequent evolution, however, differs markedly between the two
regimes. In the weak interactions regime, each dissipative channel have
a superradiant burst roughly at the same time [Fig.~\ref{fig:bursts} $(a)$, dashed curves contribute to a single peak to the total emission (solid black curve)]. Conversely, when interactions are sizable, the decay rates $\gamma_\xi = \Gamma(\Delta + \xi J)^3$ are no longer solely set by $\Delta$, leading to a marked separation of relaxation time scales. The total emission intensity has a strong short burst around $t\approx\gamma_2^{-1}$ [Fig \ref{fig:bursts} $(b)$ upper inset], and then develops
a metastable plateau on time scales comparable to $\gamma_1^{-1}$
[Fig.~\ref{fig:bursts} $(b)$ bottom inset], followed by decay on time scales given by $\gamma_0^{-1}$ [Fig. \ref{fig:bursts} $(b)$ main panel].
This distinction becomes even more pronounced in higher-dimensional
geometries, because the number of dissipative channels and their separations  {grows with} the lattice coordination number. Although the dynamics differ markedly between weak and strong interactions,
 the steady-state properties of KC superradiance are qualitatively independent
 of interaction strength. For simplicity, we therefore focus on the weakly
 interacting regime ($\Delta \gg J$) in the remainder of the paper; similar 
 conclusions would apply for $J \gtrsim \Delta$.\\
  
\textit{\textbf{Structure of the steady state}--} The coexistence of multiple dissipative channels associated with frequency-resolved emission breaks permutation invariance, in contrast to Dicke superradiance \cite{gross1982superradiance, Masson20, Masson24}, where all decay proceeds through a single symmetric channel. In the present case, dissipation becomes configuration dependent: each channel acts only on specific local spin excitation patterns  {[cf. Fig. \ref{fig:schem}]}. As  dynamics unfold, these spin configurations are selectively depleted, and the system can reach configurations that are dark to all decay channels. This {state dependent decay} prevents complete relaxation and stabilizes a non-trivial steady state with trapped excitations. To  {characterize} the structure of this steady state, we analyze the excitation dynamics {both in real and} in Fourier space. This provides direct access to how spin excitations are distributed and retained during the evolution. We introduce {the quantity $C^{\uparrow \downarrow} \equiv \frac{1}{N}\sum_{i = 1}^N\langle (1 - n_j) n_{j+1}\rangle$ detecting anti-ferromagnetic (AFM) correlations and} the Fourier-transformed lowering operator  {$\tilde{S}_k^- = \frac{1}{\sqrt{N}} \sum_j e^{- i k j} \sigma_j^-$}. The corresponding excitation number is
$\tilde n_k = \tilde{S}_k^{+}\tilde{S}_k^{-},$ with $\tilde{S}_k^{+} = (\tilde{S}_k^{-})^{\dagger}$.  {Each momentum mode captures   specific spatial patterns. For instance, the $k=0$ mode describes collective magnetization, while $k=\pi$ corresponds to staggered, antiferromagnetic patterns.} These observables have a natural interpretation as spin-wave–like excitations. 

As shown in  the inset of {Fig.~\ref{fig:QP_ent} $(a)$}, the system relaxes to a steady state with  finite excitation density. {The main panel highlights the dynamics of the AFM correlations in real space showing a finite steady state value. The inset highlights different $k$ modes dynamics. The $k=0$ mode reaches a pronounced maximum before decaying, while finite-momentum modes ($k\neq 0$) decay and then saturates towards  a nonzero plateau revealing a non-trivial steady-state.} This behavior follows from the structure of the dissipative dynamics.  Once the system reaches configurations with isolated excitations, only the $S_0^-$ channel remains active, and it   effectively behaves as a  collective lowering operator in an un-constrained system (i.e. conventional superradiance~\cite{gross1982superradiance}). However, because the earlier constrained dynamics has already driven the system out of the permutation-invariant sector, there are residual finite-momentum excitations that cannot be erased by $S_0^-$. These are the non-decaying steady state occupations plotted in the inset of {Fig.~\ref{fig:QP_ent} $(a)$}, which are 'dark' excitations decoupled from radiative processes. 


{The  dark steady-state manifold   has entanglement content.
Completeness of the local constraints implies $\sum_{\xi=0}^{2} S_{\xi}^{-} =    \sum_j \sigma_j^{-}\equiv S^{-}$,
such that any state dark to all channels in Eq.\eqref{eq:Lind} is also dark under the bare
collective lowering operator. As we show in~\cite{SM}, the
vacuum is the only separable common dark state.}
{Consequently, any   dark steady-state   that traps local excitations signaled by
 $C^{\uparrow \downarrow} > 0
$, must be entangled.  AFM correlations therefore serve as an entanglement witness, and they would be      measurable    via optical tweezers in cavity-QED realizations (see next Section).}

{To further characterize
entanglement we unravel   the open quantum dynamics into stochastic trajectories
\cite{Carmichael1993, Daley2014}. } 
For each  wave function trajectory, $\ket{\psi^s(t)}$,   with
$s =1, 2, .., M$, we compute the von Neumann entropy of half of the system 
\begin{equation}
S_{N/2}(\rho^s) = - \mathrm{Tr}\left(\rho^s_{N/2}\log \rho^s_{N/2}\right),
\end{equation}
where $\rho^s=\ket{\psi^s(t)}\bra{\psi^s(t)}$ and
$\rho^s_{N/2}=\mathrm{Tr}_{N/2}(\rho^s)$. To characterize entanglement
production, we consider the trajectory-averaged
entropy  {over an ensemble of $M$ stochastic realizations}
\begin{equation}
\overline{S_{N/2}}(t) = \frac{1}{M}\sum_{s=1}^{M} S_{N/2}(\rho^s(t)),
\end{equation}
which provides a statistical measure of bipartite entanglement generated by
the dissipative dynamics. {This unraveling corresponds to single photon counting}. 

In the main panel of {Fig.~\ref{fig:QP_ent} $(b)$} we show the trajectory-averaged half-chain {normalized} entanglement entropy  starting from a fully inverted state, comparing constrained and Dicke superradiance. Both cases display a transient entanglement peak; however, kinetically constrained superradiance yields a higher and broader maximum, indicating slower and more entangled transient dynamics. At late times, Dicke superradiance relaxes to a trivial, uncorrelated dark state, while constrained emission reaches a non-trivial steady state with finite entanglement that grows with the system size \cite{SM}. This is corroborated by the inset, which shows a bimodal distribution of single-trajectory steady-state entropies, with a peak at zero and a second peak around  {$S_{N/2}\simeq1.5$}, demonstrating a finite probability of generating  many-body correlations. {Notably, because the interaction and single-site terms commute in the Hamiltonian, Eq.~\eqref{eq:ham}, the resulting correlations are generated entirely by the emission channels in  accordance with the paradigm of dissipative state preparation}~\cite{diehl2008quantum,verstraete2009quantum}. 
These results are robust upon increasing system size. 
The   density of  excitations decoupled from radiation (in the steady state)  grow with $N$. Similar behavior is found for the statistics of trajectories supporting finite entanglement entropy (cf. {Fig \ref{fig:QP_ent} $(b)$, main plot, orange curves with varying opacity; for further details} see ~\cite{SM}).  Taken together, these results indicate that the dissipative constrained dynamics  sustain increasingly correlated quantum states for larger systems. \\

{\textit{\textbf{Prospects for a cavity-QED implementation}--} 
 Rydberg tweezer-arrays coupled to a single photonic cavity mode with decay rate $\kappa$ can provide an implementation of KC super-radiance.
The effective decay rates, $\gamma_\xi$, of the three kinetically constrained processes are proportional to    the spectral density of the photonic cavity mode evaluated at the frequency $ \omega_c-\omega_\xi $, where $\omega_c$ is the cavity detuned frequency. Such spectral density has a characteristic width proportional to  $\kappa$. 
Therefore, by tuning $\omega_c$ within the range $|\omega_c-\omega_\xi|\lesssim \kappa$,   one can control the value of the  decay rates, $\gamma_\xi$, and through the light emitted from the cavity    one can reconstruct the intensity burst plotted in Fig.~\ref{fig:bursts}.  
 We have also to satisfy the conditions $J\gtrsim \kappa$ and  
$J\gtrsim N_\xi\gamma_\xi^{\rm cav}$ in order to frequency-resolve  distinct
KC bursts; linewidths are set by the effective bright
population, $N_\xi$, of the corresponding KC radiative process
\cite{SM}. This implies a mild bound on the value of $N$, considering that current cavity-Rydbergs arrays are not yet fully scalable, and that $N_\xi$ can be significantly smaller than $N$ (the actual value depends on initial state chosen, and other non-universal aspects of dynamics).}

{A near-term route is to use direct Rydberg excitations in the recently
realized cavity-coupled $^{87}{\rm Rb}$ array of Ref.~\cite{de2026realization}.
That experiment reports $\kappa/2\pi=0.84(9)\,{\rm MHz}$ 
and a blockade radius $R_b\simeq4.8\,\mu{\rm m}$. Working slightly outside
the blockade radius,  gives   nearest
neighbour shifts $J/2\pi\sim0.7$--$2\,{\rm MHz}$ for
$d\simeq5$--$6\,\mu{\rm m}$. This is comparable to $\kappa$, opening up for the possibility to distinguish distinct KC emission channels. A complementary route would be  to employ Rydberg dressed states. As we discuss in Ref.~\cite{hosseinabadi2025pxp}, these platforms offer near-term prospects for 
$J/2\pi\simeq0.16\,{\rm MHz}$ and   $\kappa/2\pi\simeq0.1\,{\rm MHz}$, with  single-atom cooperativity   $C\simeq30$, which would further favor the possibility to resolve  different KC radiative channels.}\\

\textit{\textbf{Perspectives}--} A natural extension of our model consists in introducing an external coherent or incoherent drive field tilted with respect to the interaction direction in Eq.~\eqref{eq:ham}. This setting enables the exploration of new universality classes of kinetically constrained dissipative phase transitions. More generally, one may consider local, partially collective, or fully collective dephasing, pumping, and loss processes subject to kinetic constraints, extending our framework beyond constrained superradiance~\cite{sieberer2025universality,fazio2025many} (for a local version see for instance Ref.~\cite{Nill22}).

Related to this direction, the emergence of 'dark'   sectors of the Hilbert space, associated to the various kinetically constrained emission channels,  is suggestive of fragmentation, {and it could pave the way to dissipative counterparts of fractons, East models or scars~\cite{Turner_2018,Serbyn_2021,valencia2022kinetically,valencia2024rydberg,Moudgalya_2025,hosseinabadi2025pxp,nandkishore2019fractons,gromov2024colloquium} (for first steps towards experimental proposals see~\cite{winter2026extensive}).}

The   hierarchy of relaxational timescales generated by constrained collective decay provides a route  for   metastability. Certain many-body spin configurations can stay 'dark' and decoupled from radiation over long times (cf. plateau in inset of Fig.~\ref{fig:bursts}) opening  the possibility to protect many-body entangled states selected by dissipation, in a fashion reminiscent of  a   quantum memory. This also hints at the possibility to explore glassy-like dynamics~\cite{lubchenko2007theory} induced by cooperative emission, rather than by constrained (or disordered) coherent many body dynamics.

The possibility to implement a form of configuration-selective measurement, in which each emitted photon carries information about the local many-body spin environment of the emitter, can enable new forms of active dynamical steering beyond collective symmetric spin manifolds, which are currently at the forefront of   theoretical and experimental research efforts~\cite{schleier2010squeezing,yu2026efficient,roy2020measurement,grinkemeyer2025error,deist2022mid,cox2016deterministic,Marsh_spinglass2024,HH_short2024,HH_long2024,kroeze_spinglass2025}. Taken  together, these lines of inquiry point to a reformulation of many-body quantum optics in regimes where kinetic constraints break the intrinsic all-to-all structure of radiative processes {and are capable to promote entanglement, at variance with conventional super-radiance~\cite{qxx1-xr44,zhang2025unraveling,rosario2025unraveling}}. We aim to further advance this vision in future work. \\

\textit{\textbf{Acknowledgments}--} We acknowledge fruitful discussions with B. Buca, D. Chang, A. Clerk, J. Zeiher,  M. Stefanini, T. Thomay and R. Verresen. LFDP and JM acknowledge support from the CAS Dean's office at SUNY Buffalo.  Research was sponsored by the Army
Research Office and was accomplished under Grant Number W911NF-26-1-A176. The views and
conclusions contained in this document are those of the authors and should not be interpreted as
representing the official policies, either expressed or implied, of the Army Research Office or the U.
S. Government. The U.S. Government is authorized to reproduce and distribute reprints for
Government purposes notwithstanding any copyright notation herein”

This data was produced by State University of New York
(SUNY) at Buffalo under Army Research Office (ARO) Award Number W911NF-26-1-A176. ARO,
as the Federal awarding agency, reserves a royalty-free, nonexclusive and irrevocable right to
reproduce, publish, or otherwise use this data for Federal purposes, and to authorize others to do
so in accordance with 2 CFR 200.315(b).

The authors acknowledge the use of generative AI tools in the course of this work. AI assistance (Claude Fable 5, Anthropic) was used to independently check the entanglement witness applied to the dark states. AI assistance was also used for data readout and plotting scripts; no figure content was generated by AI. During manuscript preparation (winter 2025–summer 2026), the authors used AI tools (Claude Opus 4.8, Claude Opus 5, and Claude Fable 5, Anthropic; GPT-5.6 and earlier versions, OpenAI) for language editing, rephrasing, length control, typographical and consistency checking, and feedback on presentation and framing. AI tools were also used to identify potentially relevant prior results in the literature; all such results were checked by the authors and are cited to their original sources.
All AI-assisted calculations, checks, scripts, and outputs used in this work were independently benchmarked and validated by the authors. The physical analysis and interpretation are entirely the authors’ own. The authors take full responsibility for the content, accuracy, and conclusions of the final manuscript.\\

\textit{\textbf{Data availability}} The data and code required to reproduce the results presented in this work are available at  \cite{KCSR_zenodo}.
\bibliography{kc}

@misc{SM,
  note = {See Supplemental Material
  for additional superradiance scalings, finite-size scaling analysis,
  and derivation of the Lindblad master equation.}
}

@article{ferioli2023non,
  title={A non-equilibrium superradiant phase transition in free space},
  author={Ferioli, Giovanni and Glicenstein, Antoine and Ferrier-Barbut, Igor and Browaeys, Antoine},
  journal={Nature Physics},
  volume={19},
  number={9},
  pages={1345--1349},
  year={2023},
  publisher={Nature Publishing Group UK London}
}

@article{kim2025super_subradiant,
  author       = {Kim, Y. and Lanuza, A. and Schneble, D.},
  title        = {Super- and subradiant dynamics of quantum emitters mediated by atomic matter waves},
  journal      = {Nature Physics},
  volume       = {21},
  pages        = {70--76},
  year         = {2025},
  doi          = {10.1038/s41567-024-02676-w},
  publisher    = {Springer Nature}
}

@article{angerer2018superradiant,
  title={Superradiant emission from colour centres in diamond},
  author={Angerer, Andreas and Streltsov, Kirill and Astner, Thomas and Putz, Stefan and Sumiya, Hitoshi and Onoda, Shinobu and Isoya, Junichi and Munro, William J and Nemoto, Kae and Schmiedmayer, J{\"o}rg and others},
  journal={Nature Physics},
  volume={14},
  number={12},
  pages={1168--1172},
  year={2018},
  publisher={Nature Publishing Group UK London}
}

@article{liedl2024observation,
  title={Observation of superradiant bursts in a cascaded quantum system},
  author={Liedl, Christian and Tebbenjohanns, Felix and Bach, Constanze and Pucher, Sebastian and Rauschenbeutel, Arno and Schneeweiss, Philipp},
  journal={Physical Review X},
  volume={14},
  number={1},
  pages={011020},
  year={2024},
  publisher={APS}
}

@article{fazio2025many,
  title={Many-body open quantum systems},
  author={Fazio, Rosario and Keeling, Jonathan and Mazza, Leonardo and Schir{\`o}, Marco},
  journal={SciPost Physics Lecture Notes},
  pages={099},
  year={2025}
}

@article{sieberer2025universality,
  title={Universality in driven open quantum matter},
  author={Sieberer, Lukas M and Buchhold, Michael and Marino, Jamir and Diehl, Sebastian},
  journal={Reviews of Modern Physics},
  volume={97},
  number={2},
  pages={025004},
  year={2025},
  publisher={APS}
}

@article{kersten2026self,
	author = {Kersten, Wenzel and de Zordo, Nikolaus and Diekmann, Oliver and others},
	date = {2026/01/01},
	doi = {10.1038/s41567-025-03123-0},
	id = {Kersten2026},
	isbn = {1745-2481},
	journal = {Nature Physics},
	number = {1},
	pages = {158--163},
	title = {Self-induced superradiant masing},
	url = {https://doi.org/10.1038/s41567-025-03123-0},
	volume = {22},
	year = {2026}}

@article{hosseinabadi2025pxp,
  title = {Kinetically Constrained Cavity QED: From Blockaded Ferromagnetism to Long-Range Quantum Scars},
  author = {Hosseinabadi, Hossein and Valencia-Tortora, Riccardo J. and Mikheev, Aleksandr N. and Chang, Darrick E. and Zeiher, Johannes and Moessner, Roderich and Marino, Jamir},
  journal = {PRX Quantum},
  volume = {7},
  issue = {2},
  pages = {020357},
  numpages = {20},
  year = {2026},
  month = {Jun},
  publisher = {American Physical Society},
  doi = {10.1103/586v-sxqr},
  url = {https://link.aps.org/doi/10.1103/586v-sxqr}
}

@article{
kroeze_spinglass2025,
author = {Ronen M. Kroeze  and Brendan P. Marsh  and David Atri Schuller  and others },
title = {Directly observing replica symmetry breaking in a vector quantum-optical spin glass},
journal = {Science},
volume = {389},
number = {6765},
pages = {1122-1126},
year = {2025},
doi = {10.1126/science.adu7710},
URL = {https://www.science.org/doi/abs/10.1126/science.adu7710}}

@article{Marsh_spinglass2024,
  title = {Entanglement and Replica Symmetry Breaking in a Driven-Dissipative Quantum Spin Glass},
  author = {Marsh, Brendan P. and Kroeze, Ronen M. and Ganguli, Surya and Gopalakrishnan, Sarang and Keeling, Jonathan and Lev, Benjamin L.},
  journal = {Phys. Rev. X},
  volume = {14},
  issue = {1},
  pages = {011026},
  numpages = {24},
  year = {2024},
  month = {Feb},
  publisher = {American Physical Society},
  doi = {10.1103/PhysRevX.14.011026},
  url = {https://link.aps.org/doi/10.1103/PhysRevX.14.011026}
}

@article{HH_short2024,
  title = {Quantum-to-classical crossover in the spin glass dynamics of cavity QED simulators},
  author = {Hosseinabadi, Hossein and Chang, Darrick E. and Marino, Jamir},
  journal = {Phys. Rev. Res.},
  volume = {6},
  issue = {4},
  pages = {043313},
  numpages = {7},
  year = {2024},
  month = {Dec},
  publisher = {American Physical Society},
  doi = {10.1103/PhysRevResearch.6.043313},
  url = {https://link.aps.org/doi/10.1103/PhysRevResearch.6.043313}
}

@article{HH_long2024,
  title = {Far from equilibrium field theory for strongly coupled light and matter: Dynamics of frustrated multimode cavity QED},
  author = {Hosseinabadi, Hossein and Chang, Darrick E. and Marino, Jamir},
  journal = {Phys. Rev. Res.},
  volume = {6},
  issue = {4},
  pages = {043314},
  numpages = {37},
  year = {2024},
  month = {Dec},
  publisher = {American Physical Society},
  doi = {10.1103/PhysRevResearch.6.043314},
  url = {https://link.aps.org/doi/10.1103/PhysRevResearch.6.043314}
}

@article{diehl2008quantum,
  title={Quantum states and phases in driven open quantum systems with cold atoms},
  author={Diehl, Sebastian and Micheli, A and Kantian, Adrian and Kraus, B and B{\"u}chler, HP and Zoller, Peter},
  journal={Nature Physics},
  volume={4},
  number={11},
  pages={878--883},
  year={2008},
  publisher={Nature Publishing Group UK London}
}

@article{verstraete2009quantum,
  title={Quantum computation and quantum-state engineering driven by dissipation},
  author={Verstraete, Frank and Wolf, Michael M and Ignacio Cirac, J},
  journal={Nature physics},
  volume={5},
  number={9},
  pages={633--636},
  year={2009},
  publisher={Nature Publishing Group UK London}
}

@article{harrington2022engineered,
  title={Engineered dissipation for quantum information science},
  author={Harrington, Patrick M and Mueller, Erich J and Murch, Kater W},
  journal={Nature Reviews Physics},
  volume={4},
  number={10},
  pages={660--671},
  year={2022},
  publisher={Nature Publishing Group UK London}
}

@article{lin2013dissipative,
  title={Dissipative production of a maximally entangled steady state of two quantum bits},
  author={Lin, Yiheng and Gaebler, JP and Reiter, Florentin and Tan, Ting Rei and Bowler, Ryan and S{\o}rensen, AS and Leibfried, Dietrich and Wineland, David J},
  journal={Nature},
  volume={504},
  number={7480},
  pages={415--418},
  year={2013},
  publisher={Nature Publishing Group UK London}
}

@article{stannigel2012driven,
  title={Driven-dissipative preparation of entangled states in cascaded quantum-optical networks},
  author={Stannigel, Kai and Rabl, Peter and Zoller, Peter},
  journal={New Journal of Physics},
  volume={14},
  number={6},
  pages={063014},
  year={2012},
  publisher={IOP Publishing}
}

@article{budich2015dissipative,
  title={Dissipative preparation of Chern insulators},
  author={Budich, Jan Carl and Zoller, Peter and Diehl, Sebastian},
  journal={Physical Review A},
  volume={91},
  number={4},
  pages={042117},
  year={2015},
  publisher={APS}
}

@article{reiter2016scalable,
  title={Scalable dissipative preparation of many-body entanglement},
  author={Reiter, Florentin and Reeb, David and S{\o}rensen, Anders S},
  journal={Physical review letters},
  volume={117},
  number={4},
  pages={040501},
  year={2016},
  publisher={APS}
}

@article{plankensteiner2015selective,
  title={Selective protected state preparation of coupled dissipative quantum emitters},
  author={Plankensteiner, David and Ostermann, Laurin and Ritsch, Helmut and Genes, Claudiu},
  journal={Scientific reports},
  volume={5},
  number={1},
  pages={16231},
  year={2015},
  publisher={Nature Publishing Group UK London}
}

@article{marr2003entangled,
  title={Entangled-state preparation via dissipation-assisted adiabatic passages},
  author={Marr, Carsten and Beige, Almut and Rempe, Gerhard},
  journal={Physical Review A},
  volume={68},
  number={3},
  pages={033817},
  year={2003},
  publisher={APS}
}

@article{pocklington2025accelerating,
  title={Accelerating dissipative state preparation with adaptive open quantum dynamics},
  author={Pocklington, Andrew and Clerk, Aashish A},
  journal={Physical Review Letters},
  volume={134},
  number={5},
  pages={050603},
  year={2025},
  publisher={APS}
}

@article{kastoryano2011dissipative,
  title={Dissipative preparation of entanglement in optical cavities},
  author={Kastoryano, Michael James and Reiter, Florentin and S{\o}rensen, Anders S{\o}ndberg},
  journal={Physical review letters},
  volume={106},
  number={9},
  pages={090502},
  year={2011},
  publisher={APS}
}

@article{mi2024stable,
  title={Stable quantum-correlated many-body states through engineered dissipation},
  author={Mi, Xiao and Michailidis, AA and Shabani, Sara and Miao, KC and Klimov, PV and Lloyd, J and Rosenberg, E and Acharya, R and Aleiner, I and Andersen, TI and others},
  journal={Science},
  volume={383},
  number={6689},
  pages={1332--1337},
  year={2024},
  publisher={American Association for the Advancement of Science}
}

@article{ma2019dissipatively,
  title={A dissipatively stabilized Mott insulator of photons},
  author={Ma, Ruichao and Saxberg, Brendan and Owens, Clai and Leung, Nelson and Lu, Yao and Simon, Jonathan and Schuster, David I},
  journal={Nature},
  volume={566},
  number={7742},
  pages={51--57},
  year={2019},
  publisher={Nature Publishing Group UK London}
}

@book{breuer2002,
  title     = {The Theory of Open Quantum Systems},
  author    = {Heinz-Peter Breuer and Francesco Petruccione},
  year      = {2002},
  publisher = {Oxford University Press},
  month     = aug,
}

@article{stefanini2025lindblad,
  title={Is Lindblad for me?},
  author={Stefanini, Martino and Ziolkowska, Aleksandra A and Budker, Dmitry and Poschinger, Ulrich and Schmidt-Kaler, Ferdinand and Browaeys, Antoine and Imamoglu, Atac and Chang, Darrick and Marino, Jamir},
  journal={arXiv preprint arXiv:2506.22436},
  year={2025}
}

@article{cox2016deterministic,
  title={Deterministic squeezed states with collective measurements and feedback},
  author={Cox, Kevin C and Greve, Graham P and Weiner, Joshua M and Thompson, James K},
  journal={Physical review letters},
  volume={116},
  number={9},
  pages={093602},
  year={2016},
  publisher={APS}
}

@article{de2026realization,
  title={Realization of a cavity-coupled Rydberg array},
  author={De Santis, Jacopo and Dura-Kov{\'a}cs, Bal{\'a}zs and {\"O}nc{\"u}, Mehmet and Bouscal, Adrien and Vasileiadis, Dimitrios and Zeiher, Johannes},
  journal={arXiv preprint arXiv:2602.12152},
  year={2026}
}

@article{deist2022mid,
  title={Mid-circuit cavity measurement in a neutral atom array},
  author={Deist, Emma and Lu, Yue-Hui and Ho, Jacquelyn and Pasha, Mary Kate and Zeiher, Johannes and Yan, Zhenjie and Stamper-Kurn, Dan M},
  journal={Physical Review Letters},
  volume={129},
  number={20},
  pages={203602},
  year={2022},
  publisher={APS}
}

@article{schleier2010squeezing,
  title={Squeezing the collective spin of a dilute atomic ensemble by cavity feedback},
  author={Schleier-Smith, Monika H and Leroux, Ian D and Vuleti{\'c}, Vladan},
  journal={Physical Review A—Atomic, Molecular, and Optical Physics},
  volume={81},
  number={2},
  pages={021804},
  year={2010},
  publisher={APS}
}

@article{grinkemeyer2025error,
  title={Error-detected quantum operations with neutral atoms mediated by an optical cavity},
  author={Grinkemeyer, Brandon and Guardado-Sanchez, Elmer and Dimitrova, Ivana and Shchepanovich, Danilo and Mandopoulou, G Eirini and Borregaard, Johannes and Vuleti{\'c}, Vladan and Lukin, Mikhail D},
  journal={Science},
  volume={387},
  number={6740},
  pages={1301--1305},
  year={2025},
  publisher={American Association for the Advancement of Science}
}

@article{roy2020measurement,
  title={Measurement-induced steering of quantum systems},
  author={Roy, Sthitadhi and Chalker, JT and Gornyi, IV and Gefen, Yuval},
  journal={Physical review research},
  volume={2},
  number={3},
  pages={033347},
  year={2020},
  publisher={APS}
}

@article{yu2026efficient,
  title={Efficient preparation of Dicke states},
  author={Yu, Jeffery and Muleady, Sean R and Wang, Yu-Xin and Schine, Nathan and Gorshkov, Alexey V and Childs, Andrew M},
  journal={Physical Review Letters},
  volume={136},
  number={3},
  pages={030601},
  year={2026},
  publisher={APS}
}

@article{valencia2024rydberg,
  title={Rydberg platform for nonergodic chiral quantum dynamics},
  author={Valencia-Tortora, Riccardo J and Pancotti, Nicola and Fleischhauer, Michael and Bernien, Hannes and Marino, Jamir},
  journal={Physical Review Letters},
  volume={132},
  number={22},
  pages={223201},
  year={2024},
  publisher={APS}
}

@article{valencia2022kinetically,
  title={Kinetically constrained quantum dynamics in superconducting circuits},
  author={Valencia-Tortora, Riccardo J and Pancotti, Nicola and Marino, Jamir},
  journal={PRX Quantum},
  volume={3},
  number={2},
  pages={020346},
  year={2022},
  publisher={APS}
}

@article{gromov2024colloquium,
  title={Colloquium: fracton matter},
  author={Gromov, Andrey and Radzihovsky, Leo},
  journal={Reviews of Modern Physics},
  volume={96},
  number={1},
  pages={011001},
  year={2024},
  publisher={APS}
}

@article{nandkishore2019fractons,
  title={Fractons},
  author={Nandkishore, Rahul M and Hermele, Michael},
  journal={Annual Review of Condensed Matter Physics},
  volume={10},
  number={1},
  pages={295--313},
  year={2019},
  publisher={Annual Reviews}
}

@article{lubchenko2007theory,
  title={Theory of structural glasses and supercooled liquids},
  author={Lubchenko, Vassiliy and Wolynes, Peter G},
  journal={Annu. Rev. Phys. Chem.},
  volume={58},
  number={1},
  pages={235--266},
  year={2007},
  publisher={Annual Reviews}
}

@article{gross1982superradiance,
  title={Superradiance: An essay on the theory of collective spontaneous emission},
  author={Gross, Michel and Haroche, Serge},
  journal={Physics reports},
  volume={93},
  number={5},
  pages={301--396},
  year={1982},
  publisher={Elsevier}
}

@article{Nill22,
  title = {Many-Body Radiative Decay in Strongly Interacting Rydberg Ensembles},
  author = {Nill, Chris and Brandner, Kay and Olmos, Beatriz and Carollo, Federico and Lesanovsky, Igor},
  journal = {Phys. Rev. Lett.},
  volume = {129},
  issue = {24},
  pages = {243202},
  numpages = {6},
  year = {2022},
  month = {Dec},
  publisher = {American Physical Society},
  doi = {10.1103/PhysRevLett.129.243202},
  url = {https://link.aps.org/doi/10.1103/PhysRevLett.129.243202}
}

@book{Carmichael1993,
  author    = {Carmichael, Howard J.},
  title     = {An Open Systems Approach to Quantum Optics},
  publisher = {Springer},
  address   = {Berlin},
  year      = {1993},
}

@article{Daley2014,
  author  = {Daley, Andrew J.},
  title   = {Quantum trajectories and open many-body quantum systems},
  journal = {Advances in Physics},
  volume  = {63},
  number  = {2},
  pages   = {77--149},
  year    = {2014},
  doi     = {10.1080/00018732.2014.933502},
}

@book{NielsenChuang2010,
  author    = {Nielsen, Michael A. and Chuang, Isaac L.},
  title     = {Quantum Computation and Quantum Information},
  publisher = {Cambridge University Press},
  address   = {Cambridge},
  year      = {2010},
  edition   = {10th Anniversary Edition},
}

@article{Dicke1954,
  author  = {Dicke, Robert H.},
  title   = {Coherence in Spontaneous Radiation Processes},
  journal = {Physical Review},
  volume  = {93},
  number  = {1},
  pages   = {99–110},
  year    = {1954},
  doi     = {10.1103/PhysRev.93.99},
}

@article{Masson20,
  title = {Many-Body Signatures of Collective Decay in Atomic Chains},
  author = {Masson, Stuart J. and Ferrier-Barbut, Igor and Orozco, Luis A. and Browaeys, Antoine and Asenjo-Garcia, Ana},
  journal = {Phys. Rev. Lett.},
  volume = {125},
  issue = {26},
  pages = {263601},
  numpages = {7},
  year = {2020},
  month = {Dec},
  publisher = {American Physical Society},
  doi = {10.1103/PhysRevLett.125.263601},
  url = {https://link.aps.org/doi/10.1103/PhysRevLett.125.263601}
}

@article{Masson24,
  title = {Dicke Superradiance in Ordered Arrays of Multilevel Atoms},
  author = {Masson, Stuart J. and Covey, Jacob P. and Will, Sebastian and Asenjo-Garcia, Ana},
  journal = {PRX Quantum},
  volume = {5},
  issue = {1},
  pages = {010344},
  numpages = {19},
  year = {2024},
  month = {Mar},
  publisher = {American Physical Society},
  doi = {10.1103/PRXQuantum.5.010344},
  url = {https://link.aps.org/doi/10.1103/PRXQuantum.5.010344}
}

@article{qxx1-xr44,
  title = {Absence of entanglement growth in Dicke superradiance},
  author = {Bassler, N. S.},
  journal = {Phys. Rev. A},
  volume = {112},
  issue = {5},
  pages = {053713},
  numpages = {16},
  year = {2025},
  month = {Nov},
  publisher = {American Physical Society},
  doi = {10.1103/qxx1-xr44},
  url = {https://link.aps.org/doi/10.1103/qxx1-xr44}
}

@article{rosario2025unraveling,
  title = {Unraveling Dicke Superradiant Decay with Separable Coherent Spin States},
  author = {Rosario, P. and Solak, L. O. R. and Cidrim, A. and Bachelard, R. and Schachenmayer, J.},
  journal = {Phys. Rev. Lett.},
  volume = {135},
  issue = {13},
  pages = {133602},
  numpages = {7},
  year = {2025},
  month = {Sep},
  publisher = {American Physical Society},
  doi = {10.1103/xcxr-sm9c},
  url = {https://link.aps.org/doi/10.1103/xcxr-sm9c}
}

@article{zhang2025unraveling,
  title={Unraveling superradiance: Entanglement and mutual information in collective decay},
  author={Zhang, Xin HH and Malz, Daniel and Rabl, Peter},
  journal={Physical Review Letters},
  volume={135},
  number={3},
  pages={033602},
  year={2025},
  publisher={APS}
}

@article{winter2026extensive,
  title={Extensive mixed-state entanglement in kinetically constrained superradiance},
  author={Winter, Lucas and Kumlin, Jan and Pohl, Thomas and Nunnenkamp, Andreas},
  journal={arXiv preprint arXiv:2605.16131},
  year={2026}
}

@article{Kumlin2025,
  title = {Superradiance of Strongly Interacting Dipolar Excitons in Moir\'e Quantum Materials},
  author = {Kumlin, Jan and Srivastava, Ajit and Pohl, Thomas},
  journal = {Phys. Rev. Lett.},
  volume = {134},
  issue = {12},
  pages = {126901},
  numpages = {8},
  year = {2025},
  month = {Mar},
  publisher = {American Physical Society},
  doi = {10.1103/PhysRevLett.134.126901},
  url = {https://link.aps.org/doi/10.1103/PhysRevLett.134.126901}
}

@article{Li2025,
  title = {Solid-state platform for cooperative quantum dynamics driven by correlated emission},
  author = {Li, Xin and Marino, Jamir and Chang, Darrick E. and Flebus, Benedetta},
  journal = {Phys. Rev. B},
  volume = {111},
  issue = {6},
  pages = {064424},
  numpages = {15},
  year = {2025},
  month = {Feb},
  publisher = {American Physical Society},
  doi = {10.1103/PhysRevB.111.064424},
  url = {https://link.aps.org/doi/10.1103/PhysRevB.111.064424}
}

@article{Malz2022,
  title = {Large-$N$ limit of Dicke superradiance},
  author = {Malz, Daniel and Trivedi, Rahul and Cirac, J. Ignacio},
  journal = {Phys. Rev. A},
  volume = {106},
  issue = {1},
  pages = {013716},
  numpages = {15},
  year = {2022},
  month = {Jul},
  publisher = {American Physical Society},
  doi = {10.1103/PhysRevA.106.013716},
  url = {https://link.aps.org/doi/10.1103/PhysRevA.106.013716}
}

@article{Lambert2016,
  title = {Superradiance with an ensemble of superconducting flux qubits},
  author = {Lambert, Neill and Matsuzaki, Yuichiro and Kakuyanagi, Kosuke and Ishida, Natsuko and Saito, Shiro and Nori, Franco},
  journal = {Phys. Rev. B},
  volume = {94},
  issue = {22},
  pages = {224510},
  numpages = {8},
  year = {2016},
  month = {Dec},
  publisher = {American Physical Society},
  doi = {10.1103/PhysRevB.94.224510},
  url = {https://link.aps.org/doi/10.1103/PhysRevB.94.224510}
}

@article{Clopez2023,
  title = {Many-Body Superradiance and Dynamical Mirror Symmetry Breaking in Waveguide QED},
  author = {Cardenas-Lopez, Silvia and Masson, Stuart J. and Zager, Zoe and Asenjo-Garcia, Ana},
  journal = {Phys. Rev. Lett.},
  volume = {131},
  issue = {3},
  pages = {033605},
  numpages = {7},
  year = {2023},
  month = {Jul},
  publisher = {American Physical Society},
  doi = {10.1103/PhysRevLett.131.033605},
  url = {https://link.aps.org/doi/10.1103/PhysRevLett.131.033605}
}

@article{Lohof2023,
  title = {Signatures of Superradiance as a Witness to Multipartite Entanglement},
  author = {Lohof, Frederik and Schumayer, Daniel and Hutchinson, David A. W. and Gies, Christopher},
  journal = {Phys. Rev. Lett.},
  volume = {131},
  issue = {6},
  pages = {063601},
  numpages = {6},
  year = {2023},
  month = {Aug},
  publisher = {American Physical Society},
  doi = {10.1103/PhysRevLett.131.063601},
  url = {https://link.aps.org/doi/10.1103/PhysRevLett.131.063601}
}

@article{Windt2025,
  title = {Effects of Retardation on Many-Body Superradiance in Chiral Waveguide QED},
  author = {Windt, Bennet and Bello, Miguel and Malz, Daniel and Cirac, J. Ignacio},
  journal = {Phys. Rev. Lett.},
  volume = {134},
  issue = {17},
  pages = {173601},
  numpages = {8},
  year = {2025},
  month = {Apr},
  publisher = {American Physical Society},
  doi = {10.1103/PhysRevLett.134.173601},
  url = {https://link.aps.org/doi/10.1103/PhysRevLett.134.173601}
}

@article{RBigorda2022,
  title = {Superradiance and subradiance in inverted atomic arrays},
  author = {Rubies-Bigorda, Oriol and Yelin, Susanne F.},
  journal = {Phys. Rev. A},
  volume = {106},
  issue = {5},
  pages = {053717},
  numpages = {12},
  year = {2022},
  month = {Nov},
  publisher = {American Physical Society},
  doi = {10.1103/PhysRevA.106.053717},
  url = {https://link.aps.org/doi/10.1103/PhysRevA.106.053717}
}

@article{Kirton2017,
  title = {Suppressing and Restoring the Dicke Superradiance Transition by Dephasing and Decay},
  author = {Kirton, Peter and Keeling, Jonathan},
  journal = {Phys. Rev. Lett.},
  volume = {118},
  issue = {12},
  pages = {123602},
  numpages = {5},
  year = {2017},
  month = {Mar},
  publisher = {American Physical Society},
  doi = {10.1103/PhysRevLett.118.123602},
  url = {https://link.aps.org/doi/10.1103/PhysRevLett.118.123602}
}

@article{Andolina2024,
  title = {Dicke superradiant heat current enhancement in circuit quantum electrodynamics},
  author = {Andolina, Gian Marcello and Erdman, Paolo Andrea and No\'e, Frank and Pekola, Jukka and Schir\`o, Marco},
  journal = {Phys. Rev. Res.},
  volume = {6},
  issue = {4},
  pages = {043128},
  numpages = {17},
  year = {2024},
  month = {Nov},
  publisher = {American Physical Society},
  doi = {10.1103/PhysRevResearch.6.043128},
  url = {https://link.aps.org/doi/10.1103/PhysRevResearch.6.043128}
}

@article{DeVoe_SR1996,
  title = {Observation of Superradiant and Subradiant Spontaneous Emission of Two Trapped Ions},
  author = {DeVoe, R. G. and Brewer, R. G.},
  journal = {Phys. Rev. Lett.},
  volume = {76},
  issue = {12},
  pages = {2049--2052},
  numpages = {0},
  year = {1996},
  month = {Mar},
  publisher = {American Physical Society},
  doi = {10.1103/PhysRevLett.76.2049},
  url = {https://link.aps.org/doi/10.1103/PhysRevLett.76.2049}
}

@article{Guerin_SR2016,
  title = {Subradiance in a Large Cloud of Cold Atoms},
  author = {Guerin, William and Ara\'ujo, Michelle O. and Kaiser, Robin},
  journal = {Phys. Rev. Lett.},
  volume = {116},
  issue = {8},
  pages = {083601},
  numpages = {5},
  year = {2016},
  month = {Feb},
  publisher = {American Physical Society},
  doi = {10.1103/PhysRevLett.116.083601},
  url = {https://link.aps.org/doi/10.1103/PhysRevLett.116.083601}
}

@article{Goban_WaveGuideSR2015,
  title = {Superradiance for Atoms Trapped along a Photonic Crystal Waveguide},
  author = {Goban, A. and Hung, C.-L. and Hood, J. D. and Yu, S.-P. and Muniz, J. A. and Painter, O. and Kimble, H. J.},
  journal = {Phys. Rev. Lett.},
  volume = {115},
  issue = {6},
  pages = {063601},
  numpages = {5},
  year = {2015},
  month = {Aug},
  publisher = {American Physical Society},
  doi = {10.1103/PhysRevLett.115.063601},
  url = {https://link.aps.org/doi/10.1103/PhysRevLett.115.063601}
}

@article{Solano_WaveGuideSR2016,
	author = {Solano, P. and Barberis-Blostein, P. and Fatemi, F. K. and Orozco, L. A. and Rolston, S. L.},
	date = {2017/11/30},
	doi = {10.1038/s41467-017-01994-3},
	id = {Solano2017},
	isbn = {2041-1723},
	journal = {Nature Communications},
	number = {1},
	pages = {1857},
	title = {Super-radiance reveals infinite-range dipole interactions through a nanofiber},
	url = {https://doi.org/10.1038/s41467-017-01994-3},
	volume = {8},
	year = {2017}}

@article{Sheremet_WaveGuide2023,
  title = {Waveguide quantum electrodynamics: Collective radiance and photon-photon correlations},
  author = {Sheremet, Alexandra S. and Petrov, Mihail I. and Iorsh, Ivan V. and Poshakinskiy, Alexander V. and Poddubny, Alexander N.},
  journal = {Rev. Mod. Phys.},
  volume = {95},
  issue = {1},
  pages = {015002},
  numpages = {59},
  year = {2023},
  month = {Mar},
  publisher = {American Physical Society},
  doi = {10.1103/RevModPhys.95.015002},
  url = {https://link.aps.org/doi/10.1103/RevModPhys.95.015002}
}

@article{
Tiranov_WaveGuideSR2023,
author = {Alexey Tiranov  and Vasiliki Angelopoulou  and Cornelis Jacobus van Diepen  and Björn Schrinski  and Oliver August Dall’Alba Sandberg  and Ying Wang  and Leonardo Midolo  and Sven Scholz  and Andreas Dirk Wieck  and Arne Ludwig  and Anders Søndberg Sørensen  and Peter Lodahl },
title = {Collective super- and subradiant dynamics between distant optical quantum emitters},
journal = {Science},
volume = {379},
number = {6630},
pages = {389-393},
year = {2023},
doi = {10.1126/science.ade9324},
URL = {https://www.science.org/doi/abs/10.1126/science.ade9324}}

@article{Pallmann_NVSR2024,
  title = {Cavity-Mediated Collective Emission from Few Emitters in a Diamond Membrane},
  author = {Pallmann, Maximilian and K\"oster, Kerim and Zhang, Yuan and Heupel, Julia and Eichhorn, Timon and Popov, Cyril and M\o{}lmer, Klaus and Hunger, David},
  journal = {Phys. Rev. X},
  volume = {14},
  issue = {4},
  pages = {041055},
  numpages = {14},
  year = {2024},
  month = {Dec},
  publisher = {American Physical Society},
  doi = {10.1103/PhysRevX.14.041055},
  url = {https://link.aps.org/doi/10.1103/PhysRevX.14.041055}
}

@article{Masson_Universality2022,
	author = {Masson, Stuart J. and Asenjo-Garcia, Ana},
	date = {2022/04/27},
	doi = {10.1038/s41467-022-29805-4},
	id = {Masson2022},
	isbn = {2041-1723},
	journal = {Nature Communications},
	number = {1},
	pages = {2285},
	title = {Universality of Dicke superradiance in arrays of quantum emitters},
	url = {https://doi.org/10.1038/s41467-022-29805-4},
	volume = {13},
	year = {2022}}

@article{Moudgalya_2025,
  title = {Hilbert Space Fragmentation and Commutant Algebras},
  author = {Moudgalya, Sanjay and Motrunich, Olexei I.},
  journal = {Phys. Rev. X},
  volume = {12},
  issue = {1},
  pages = {011050},
  numpages = {44},
  year = {2022},
  month = {Mar},
  publisher = {American Physical Society},
  doi = {10.1103/PhysRevX.12.011050},
  url = {https://link.aps.org/doi/10.1103/PhysRevX.12.011050}
}

@article{Serbyn_2021,
	author = {Serbyn, Maksym and Abanin, Dmitry A. and Papi{\'c}, Zlatko},
	date = {2021/06/01},
	doi = {10.1038/s41567-021-01230-2},
	id = {Serbyn2021},
	isbn = {1745-2481},
	journal = {Nature Physics},
	number = {6},
	pages = {675--685},
	title = {Quantum many-body scars and weak breaking of ergodicity},
	url = {https://doi.org/10.1038/s41567-021-01230-2},
	volume = {17},
	year = {2021}}

@article{Turner_2018,
	author = {Turner, C. J. and Michailidis, A. A. and Abanin, D. A. and Serbyn, M. and Papi{\'c}, Z.},
	date = {2018/07/01},
	doi = {10.1038/s41567-018-0137-5},
	id = {Turner2018},
	isbn = {1745-2481},
	journal = {Nature Physics},
	number = {7},
	pages = {745--749},
	title = {Weak ergodicity breaking from quantum many-body scars},
	url = {https://doi.org/10.1038/s41567-018-0137-5},
	volume = {14},
	year = {2018}}

@article{Agarwal_directional2024,
  title = {Directional Superradiance in a Driven Ultracold Atomic Gas in Free Space},
  author = {Agarwal, Sanaa and Chaparro, Edwin and Barberena, Diego and Orioli, A. Pi\~neiro and Ferioli, G. and Pancaldi, S. and Ferrier-Barbut, I. and Browaeys, A. and Rey, A.M.},
  journal = {PRX Quantum},
  volume = {5},
  issue = {4},
  pages = {040335},
  numpages = {31},
  year = {2024},
  month = {Dec},
  publisher = {American Physical Society},
  doi = {10.1103/PRXQuantum.5.040335},
  url = {https://link.aps.org/doi/10.1103/PRXQuantum.5.040335}
}

@article{Sinha_nonMarkov2020,
  title = {Non-Markovian Collective Emission from Macroscopically Separated Emitters},
  author = {Sinha, Kanupriya and Meystre, Pierre and Goldschmidt, Elizabeth A. and Fatemi, Fredrik K. and Rolston, S. L. and Solano, Pablo},
  journal = {Phys. Rev. Lett.},
  volume = {124},
  issue = {4},
  pages = {043603},
  numpages = {7},
  year = {2020},
  month = {Jan},
  publisher = {American Physical Society},
  doi = {10.1103/PhysRevLett.124.043603},
  url = {https://link.aps.org/doi/10.1103/PhysRevLett.124.043603}
}

@misc{KCSR_zenodo,
  author    = {dos Prazeres, Luis Fernando and Hosseinabadi, Hossein and Marino, Jamir},
  title     = {{Dataset and plotting scripts for "Kinetically constrained superradiance"}},
  year      = {2026},
  publisher = {Zenodo},
  doi       = {10.5281/zenodo.21225843},
  url       = {https://doi.org/10.5281/zenodo.21225843}
}

\clearpage
\onecolumngrid

\clearpage
\onecolumngrid

\begin{center}
{\large \textbf{Supplemental Material}}\\[0.5em]
{\large \textbf{Kinetically constrained superradiance}}\\[1.0em]

\end{center}

\vspace{1em}

\setcounter{equation}{0}
\renewcommand{\theequation}{S\arabic{equation}}

\setcounter{figure}{0}
\renewcommand{\thefigure}{S\arabic{figure}}

\section{Microscopic derivation}

We consider a one-dimensional array of $N$ two-level atoms at fixed positions,
interacting via nearest-neighbor interactions with periodic boundary conditions and coupled to the
electromagnetic vacuum. The total Hamiltonian is
\begin{equation}
H = H_A + H_{EM} + H_I ,
\end{equation}
with
\begin{equation}
H_A = \Delta\sum_{i=1}^N n_i
+ J \sum_{\langle ij\rangle} n_i n_j,
\qquad
H_{EM} = \sum_{k,\lambda} \omega_k b^\dagger_{k, \lambda} b_{k, \lambda} \ \ \ \ \text{and}
\qquad
H_I = -\sum_{i=1}^N \mathbf{d}\cdot\mathbf{E}(\mathbf{r}_i)\,\sigma_i^- + \mathbf{d^*}\cdot\mathbf{E}(\mathbf{r}_i)\,\sigma_i^+,
\end{equation}
where $n_i=(1+\sigma_i^z)/2$, $b^\dagger_{k,\lambda}$ ($b_{k,\lambda}$) creates (annihilates) a photon of frequency $\omega_k$ and polarization $\lambda$, and $\mathbf d$ is the vectorial electric-dipole transition matrix element, whose magnitude sets the coupling strength while its direction determines the polarization component of the electromagnetic field coupled to the atomic transition. The field is treated as a stationary bath at $T=0$, with initial state $\rho_{EM}$.

\

\emph{Born--Markov approximation.}
We aim to eliminate the electromagnetic degrees of freedom with standard Lindbladian dynamics assumptions \cite{breuer2002, stefanini2025lindblad}. 

\

Let $\rho(t)$ be the density operator of system+field. In the interaction picture
with respect to $H_0=H_A+H_{EM}$,
\begin{equation}
\dot{\rho}^I(t) = -i[H_I(t),\rho^I(t)],\qquad
H_I(t)=e^{iH_0 t}H_I e^{-iH_0 t}.
\end{equation}

The atomic part of the density matrix is obtained by tracing out the electromagnetic (EM) degrees of freedom. We assume weak coupling between the system and the EM field and treat the latter as a reservoir (Born approximation). Furthermore, we assume that the reservoir correlation time $\tau_B$ is short compared to the characteristic timescale of the atomic degrees of freedom, leading to a memoryless dynamics (Markov approximation). We then decompose the atomic part of the interaction Hamiltonian into eigenoperators of $H_A$

\begin{equation}\label{eq:eigop}
    A_i(\omega_\xi) = \sum_{\epsilon - \epsilon^\prime = \omega_\xi} \Pi_\epsilon \sigma_j^- \Pi_{\epsilon^\prime},
\end{equation}
with $\Pi_{\epsilon}$ is the projector onto the eigenspace of energy $\epsilon$ summed all over the subspaces with energy difference $\omega_\xi$. These operators correspond to the ladder operators of the free atomic Hamiltonian. Atomic interactions $(J)$ make the local transition frequency depending explicitly on the neighboring configuration. 
In one dimension, with the spectral structure of $H_A$ the energy cost of a local de-excitation at site $i$ depend on the neighboring occupations 
\begin{equation}
\xi=n_{i-1}+n_{i+1}\in\{0,1,2\}.
\end{equation}
Then, the transition
$\ket{a,\uparrow,b} \to \ket{a,\downarrow,b} $ have will emit with frequency
\begin{equation}
\omega_\xi=\Delta+\xi J,
\end{equation}
with $\xi=0,1,2$ depending on the states $\ket{a}, \ket{b}$. Therefore, we just need to be restricted to three-sites configurations, and the eigenoperators in Eq. \eqref{eq:eigop} found algebraically through the operator identity (see Ref. \cite{breuer2002} for more details)
\begin{equation}\label{eq:comm}
[H_A,A_i(\omega_\xi)] = - \omega_\xi A(\omega_\xi).
\end{equation}
Solving Eq. \eqref{eq:comm}, we find
\begin{align}\label{eq:jumps}
A_i(\omega_0) &= (1-n_{i-1})\sigma_i^-(1-n_{i+1}), \nonumber\\
A_i(\omega_1) &= n_{i-1}\sigma_i^-(1-n_{i+1})
+(1-n_{i-1})\sigma_i^- n_{i+1}, \\
A_i(\omega_2) &= n_{i-1}\sigma_i^- n_{i+1}. \nonumber
\end{align}
The interaction $(J)$ thus splits the bare transition $(\Delta)$ into three spectrally resolved channels associated with the local excitation number $\xi$, providing the microscopic origin of state-dependent decay.

\

These assumptions, transformations and approximations leads to the Redfield equation
\begin{align}
\dot\rho_S^I(t)
&=
\sum_{\omega,\omega'}\sum_{i,j}
e^{i(\omega-\omega')t}\Gamma_{ij}(\omega)
\Big[
A_j(\omega)\rho_S^I(t) A_i^\dagger(\omega')
-
A_i^\dagger(\omega')A_j(\omega)\rho_S^I(t)
\Big]
+ \text{H.c.},
\label{eq:redfield}
\end{align}
with bath correlation functions
\begin{equation}
\Gamma_{ij}(\omega)=\int_0^\infty ds\, e^{i\omega s}\,
\langle B_i(s)B_j(0)\rangle_{EM},
\qquad
\langle \cdots\rangle_{EM}\equiv \Tr_{EM}(\cdots \rho_{EM}).
\label{eq:Gamma_def}
\end{equation}
A Lamb-shift Hamiltonian arises from $\Im \Gamma_{ij}(\omega)$ and is omitted below. 

\

We integrate Eq. \eqref{eq:Gamma_def} assuming emission into free-space and subwavelength arrangement of the atoms, i.e the atoms are confined within a spatial scale way smaller than its emitted wavelength \cite{breuer2002}. The real part of Eq. \eqref{eq:Gamma_def} is the frequency dependent decay rate (in natural units), 
$\gamma(\omega)=4/3|\mathbf{d}|^2\omega^3$. For a local version of this model see Ref. \cite{Nill22}.

\

\emph{Rotating-wave approximation --} The Redfield generator in Eq.~\eqref{eq:redfield} contains oscillatory terms of the form $e^{i(\omega-\omega')t}$, which couple transitions with different Bohr frequencies. In our case, the transition frequencies differ by interaction-induced shifts $|\omega - \omega^\prime| =|  \xi-\xi'|J$, so that distinct decay channels are separated by multiples of $J$. When these splittings are large compared to the dissipative rates ($J \gg \gamma(\omega)$), the rapidly oscillating terms with $\omega\neq\omega'$ average to zero on the timescale of the system dynamics. Retaining only the resonant contributions with $\omega=\omega'$ eliminates cross-frequency couplings. Note that, in the absence of interactions ($J = 0$) Eq. \eqref{eq:redfield} subject to a subwavelength arrangement of atoms is reduced to Dicke superradiance, for more details see \cite{breuer2002}.
{While the microscopic derivation of the master equation is performed at the single-emitter level,   the decomposition into independent frequency-resolved Lindblad channels is ultimately ruled by the linewidth of the corresponding collective decay process. Since a superradiant burst has a characteristic linewidth of order $N_\xi \gamma_\xi$, where $N_\xi$ denotes the effective number of emitters participating in the bright sector of channel $\xi$, the secular condition is more generally
$J \gg N_\xi \gamma_\xi$.
In our model, however, each collective jump operator ($\hat{L}_\xi$) acts only on the subset of atoms satisfying the corresponding kinetic constraint. As a result, only this constrained bright sector contributes  to the collective linewidth, and the relevant number of emitters is the effective bright population $N_\xi$, which is generally expected to be substantially smaller than the total number of atoms.}  \\

  After RWA, the master equation therefore reduces to a sum of three independent Lindblad dissipators,
\begin{equation}
\dot\rho_S =
\sum_{\xi=0}^2 \gamma_\xi
\left[
S^-_\xi\rho_S S^+_\xi
-\frac{1}{2}\left\{S^+_\xi S^-_\xi,\rho_S\right\}
\right],
\label{eq:lind_constrained}
\end{equation}
where $S^-_\xi$ governs decay processes at frequency $\omega_\xi=\Delta+\xi J$.

Equation~\eqref{eq:lind_constrained} was written for the interaction-picture
state $\rho_S^I(t)=e^{iH_A t}\rho_S(t)e^{-iH_A t}$. Returning to the
Schr\"odinger picture gives
\begin{equation}
\dot\rho_S(t) = -i[H_A,\rho_S(t)] + \sum_{\xi=0}^2 \gamma_\xi
\left[
S^-_\xi\rho_S S^+_\xi
-\frac{1}{2}\left\{S^+_\xi S^-_\xi,\rho_S\right\}
\right]
\label{eq:ME_S}
\end{equation}
the master equation (Eq. $(2)$ in the main text) governing kinetically constrained superradiance. A similar derivation can be found in \cite{Nill22} where the authors study the effect of interactions on spontaneous emission away from the collective regime.

\section{Superradiant burst scalings}
\begin{figure*}[t]
    \centering
\includegraphics[width=0.49\textwidth]{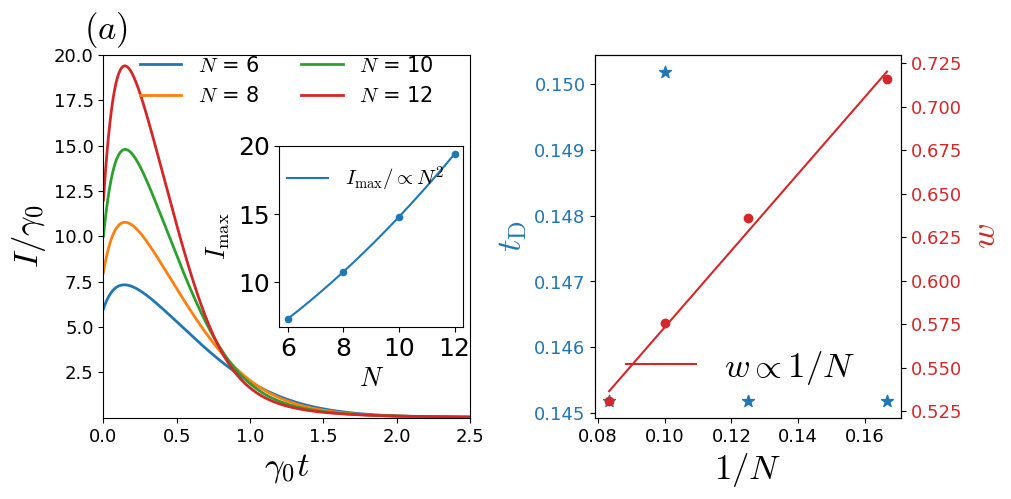}\hfill
\includegraphics[width=0.49\textwidth]{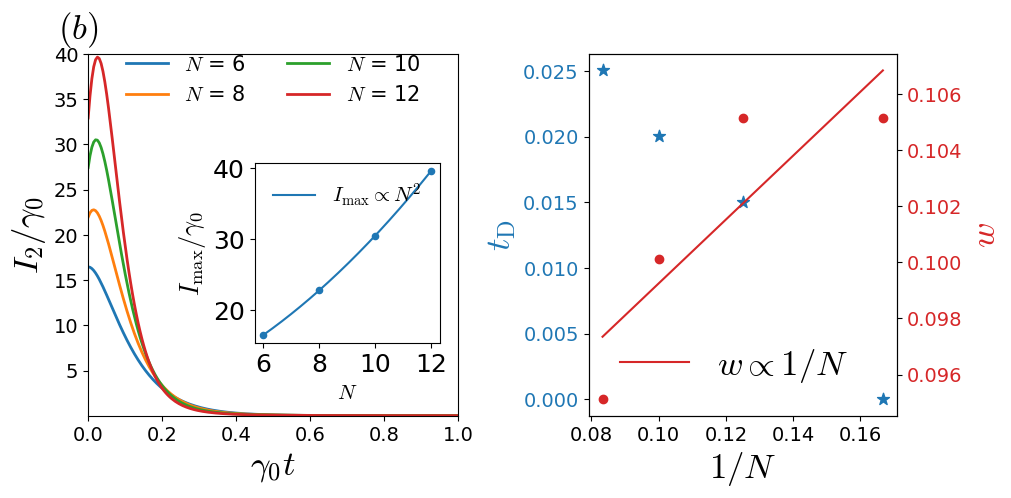}
    \vspace{2mm}
\includegraphics[width=0.49\textwidth]{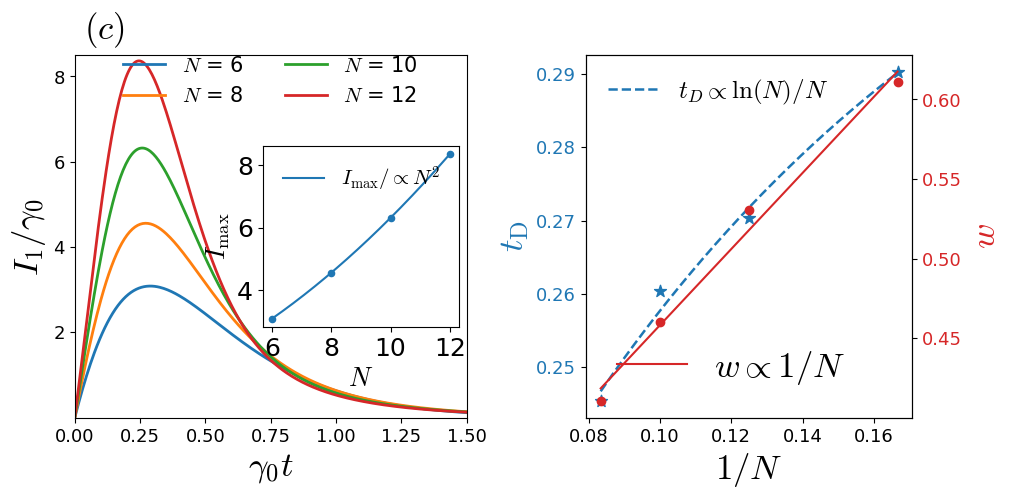}\hfill
\includegraphics[width=0.49\textwidth]{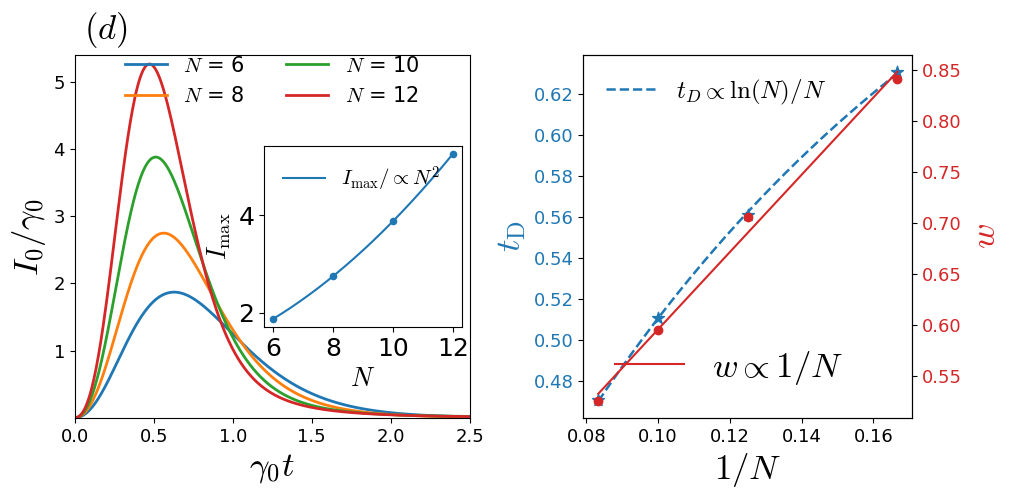}
    \caption{Finite-size scaling of superradiant bursts in the kinetically
constrained emission. Panels (a)–(d) show, respectively, the system-size
dependence for the total emission $I(t)$ (a), and for the
frequency-resolved components $I_2(t)$ (b), $I_1(t)$ (c), and $I_0(t)$ (d),
within the weak interaction regime $(\Delta/J = 0.2)$ in arbitrary units. Shown are the peak
intensity $(I_{\text{max}})$, burst width $(w)$, and delay time $(t_D)$ extracted from each emission
channel. The polynomial fit (red line) is performed exclusively for $w$
(red markers) while the $\ln (N)/N$ fit is done only for $I_0(t)$ and $I_1(t)$ and marked as a blue dashed line.}
\label{fig:sm_burst_scalings}
\end{figure*}
Superradiant bursts are known to exhibit characteristic scaling laws. The most prominent one is the scaling of the peak emission rate with the number of emitters, $I_{\max} \propto N^2$ \cite{breuer2002, gross1982superradiance, Masson20, Masson24}. In addition, the burst duration—quantified by the width of the superradiant peak—scales as $w \propto 1/N$ \cite{breuer2002, gross1982superradiance}. These scalings are robust signatures of collective emission.

They can be understood within a simple description of collective decay. Consider an ensemble of $N$ identical two-level emitters confined within a subwavelength volume, such that the spatial extent of the system is much smaller than the emitted wavelength. In this regime, the dynamics is permutation invariant and remains confined to the Dicke manifold, spanned by collective states $\ket{J,M}$ labeled by the total angular momentum $J$ and its magnetization $M$. The collective lowering operator $S^-=\sum_{j=1}^N \sigma_j^-$ acts as
\begin{equation}
S^- \ket{J,M} = \sqrt{(J+M)(J-M+1)}\, \ket{J,M-1}.
\label{eq:dicke_action}
\end{equation}

In the absence of dephasing, the reduced density matrix obeys the Lindblad master equation
\begin{equation}
\dot{\rho}=\gamma\left(S^- \rho S^+ - \tfrac{1}{2}\{S^+S^-,\rho\}\right).
\label{eq:dicke_lindblad}
\end{equation}
The above scalings directly follow from the $N$ dependence of the emission rate 
$I=\gamma \langle S^+S^- \rangle$. Using Eq.~\eqref{eq:dicke_action}, one obtains
\begin{equation}
\bra{J,M} S^+S^- \ket{J,M} = (J+M)(J-M+1),
\label{eq:intensity_dicke}
\end{equation}
which, for the fully symmetric manifold $J=N/2$, attains a maximum that scales as 
$I_{\max}\propto N^2$.

The peak emission occurs at $M=0$, corresponding to half of the emitters having decayed. 
The characteristic burst width can then be estimated by comparing the change in 
collective polarization to the maximal emission rate,
\begin{equation}
w \sim \frac{\Delta S_z}{I_{\max}} \sim \frac{1}{\gamma N},
\end{equation}
where $\Delta S_z = \langle S_z \rangle_{M=0} - \langle S_z \rangle_{M=N/2}$ and 
$S_z = \sum_{j=1}^N \sigma_j^z$.

A simple qualitative estimate for the superradiant delay time can be obtained following the argument of Gross and Haroche~\cite{gross1982superradiance}. In the Dicke ladder picture, the emission process is viewed as a cascade of collective states $\ket{J,M}$, where the transition rate between successive states is $\Gamma_{M\rightarrow M-1}=\Gamma (J+M)(J-M+1)$. At the beginning of the dynamics, starting from a fully inverted state, only a small number $s=J-M\ll N$ of photons has been emitted, so that one can approximate $J+M\simeq N$ and $J-M+1\simeq s$, yielding a linearized rate $\Gamma_s \simeq \Gamma N s$. The average waiting time for the $s$-th emission event is therefore $(\Gamma N s)^{-1}$, and summing these elementary intervals up to order $N$ gives an estimate for the buildup (delay) time of the superradiant burst, $t_D \sim \sum_{s=1}^{N}(\Gamma N s)^{-1} \simeq (\log N)/(\Gamma N)$, which captures the characteristic logarithmic correction to the cooperative $1/N$ scaling of the delay time.

\

Kinetically constrained superradiance displays the same scalings as conventional superradiance. These behaviors hold for both the total emission \( I(t) = \sum_{\xi = 0}^2 I_\xi \) and the 
frequency-resolved components \( I_\xi(t) = \gamma_\xi \langle S^+_\xi S^-_\xi \rangle \) with \( \xi = 2,1,0 \).  Importantly, this scaling remains unchanged even in the presence of the 
time-scale separation characteristic of the strongly interacting regime 
(\( J \gtrsim \Delta \)).

\

Figure~\ref{fig:sm_burst_scalings} displays the system-size dependence of the peak intensity and characteristic time scales for each emission channel for an initially fully inverted state $\ket{\uparrow}^{\otimes N}$. The $\xi=2$ component exhibits the strongest finite-size sensitivity. For small $N$ the bright sector addressed by $S^-_2$ is too restricted for the burst peak to build up sufficiently above its early-time value. Consequently, the peak formation is weak and the extracted delay time does not display the clear system-size dependence observed for $\xi=0,1$, appearing instead as scattered points (cf. Fig. \ref{fig:sm_burst_scalings} $(a)$ and $(b)$) with comparatively small values. Since the early-time dynamics of the total emission $I(t)$ is governed predominantly by $I_2(t)$, this finite-size limitation propagates to the total delay time $t_D$, producing a similarly non-monotonic and scattered behavior at small $N$. For this reason, we do not perform fits of $t_D$ for either the $\xi=2$ channel or the total emission. Aside from these delay-time finite-size effects, the remaining extracted scalings are consistent with standard superradiant behavior.

\section{Dark state entanglement}
{The presence of anti-ferromagnetic correlations $\langle \overline{n}_j n_{j +1} \rangle$ in the dark states of the kinetically constrained jump operators signals non-separability of the density matrix. We argue that it is not possible to have a separable density matrix with finite AFM while being dark to jump operators $S^-_\xi$ for $\xi = 0, 1, 2$. In this section we prove this claim.}

\begin{theorem}
Consider the kinetically constrained collective jump operators
\begin{equation}
    S^-_\xi = \sum_j P_j^{(\xi)} \sigma_j^-, \qquad \xi = 0, 1, 2
\end{equation}
where $P_j^{(0)} = (1-n_{j-1})(1-n_{j+1})$, 
$P_j^{(1)} = n_{j-1}(1-n_{j+1}) + (1-n_{j-1})n_{j+1}$ and $P_j^{(2)} = n_{j-1}n_{j+1}$.
Let $\rho_D$ be a dark state satisfying $\langle S^+_\xi S^-_\xi \rangle_{\rho_D} = 0$ 
for $\xi = 0, 1, 2$. If there exists a site $j$ such that
\begin{equation}
   \langle W\rangle_{\rho_D} \equiv  \langle (1 - n_{j-1})\, n_j \rangle_{\rho_D} > 0,
\end{equation}
then $\rho_D$ is entangled. Equivalently, every separable dark state satisfies 
$\langle (1-n_{j-1})\,n_j\rangle = 0$ for all $j$.
\end{theorem}

We first record two facts.

\begin{lemma}[Channel completeness]\label{lem:complete}
At every site $\sum_{\xi=0}^{2}P_j^{(\xi)}= 1$, and consequently
\begin{equation}
  \sum_{\xi=0}^{2}S^-_\xi=\sum_j\Big(\sum_{\xi}P_j^{(\xi)}\Big)\sigma_j^-
  =\sum_j\sigma_j^-=:S^- ,
\end{equation}
the bare collective lowering operator.
\end{lemma}
\begin{proof}
The two neighbours of $j$ carry $0$, $1$, or $2$ excitations, and
$P_j^{(0)},P_j^{(1)},P_j^{(2)}$ are the orthogonal projectors onto these three
mutually exclusive and exhaustive sectors; they therefore sum to the identity on
the neighbour pair and trivially on every other tensor factor.
\end{proof}

\begin{lemma}[Product dark states are trivial]\label{lem:prod}
For any product state $|\Phi\rangle=\bigotimes_k\big(\alpha_k|0\rangle+\beta_k|1\rangle\big)$,
\begin{equation}\label{eq:norm}
  \big\|S^-|\Phi\rangle\big\|^2
  =\sum_k|\beta_k|^4+\Big|\sum_k\bar\alpha_k\beta_k\Big|^2 .
\end{equation}
In particular $S^-|\Phi\rangle=0$ iff $\beta_k=0$ for all $k$, i.e.\ iff
$|\Phi\rangle=|0\rangle^{\otimes N}$.
\end{lemma}
\begin{proof}
Put $z_k=\bar\alpha_k\beta_k=\langle\sigma_k^-\rangle$, so $\langle\sigma_k^+\rangle=\bar z_k$
and $\langle n_k\rangle=|\beta_k|^2$. As the state is a product,
$\langle\sigma_j^+\sigma_k^-\rangle=\bar z_jz_k$ for $j\neq k$ and
$\langle\sigma_j^+\sigma_j^-\rangle=|\beta_j|^2$. Hence
\begin{align}
  \big\|S^-|\Phi\rangle\big\|^2
  &=\sum_{j,k}\langle\sigma_j^+\sigma_k^-\rangle
   =\sum_j|\beta_j|^2+\sum_{j\neq k}\bar z_jz_k \notag\\
  &=\sum_j|\beta_j|^2-\sum_j|z_j|^2+\Big|\sum_k z_k\Big|^2
   =\sum_j|\beta_j|^2\big(1-|\alpha_j|^2\big)+\Big|\sum_k z_k\Big|^2 ,
\end{align}
and $1-|\alpha_j|^2=|\beta_j|^2$ yields \eqref{eq:norm}. Both terms are
non-negative, so the sum vanishes iff every $|\beta_k|^4=0$.
\end{proof}

\begin{proof}[Proof of the Theorem]
Assume for contradiction that $\rho_D$ is separable,
\begin{equation}
  \rho_D=\sum_\mu p_\mu|\Phi_\mu\rangle\langle\Phi_\mu|,\qquad p_\mu>0,\quad
  |\Phi_\mu\rangle=\bigotimes_k|\phi_k^\mu\rangle .
\end{equation}
\\

For each $\xi$, $S^+_\xi S^-_\xi\ge0$, so
\begin{equation}
  0=\langle S^+_\xi S^-_\xi\rangle_{\rho_D}
   =\sum_\mu p_\mu\,\big\|S^-_\xi|\Phi_\mu\rangle\big\|^2 .
\end{equation}
Each term is non-negative and $p_\mu>0$, hence $S^-_\xi|\Phi_\mu\rangle=0$ for all
$\mu$ and all $\xi\in\{0,1,2\}$.\\

Summing over $\xi$ and using Lemma~\ref{lem:complete},
\begin{equation}
  S^-|\Phi_\mu\rangle=\sum_{\xi}S^-_\xi|\Phi_\mu\rangle=0\qquad\forall\,\mu .
\end{equation}
\\ 

Each $|\Phi_\mu\rangle$ is a product state with $S^-|\Phi_\mu\rangle=0$, so by
Lemma~\ref{lem:prod}, $|\Phi_\mu\rangle=|0\rangle^{\otimes N}$ for every $\mu$.
Therefore $\rho_D=|0\rangle^{\otimes N}\langle0|^{\otimes N}$: the vacuum is the
\emph{only} separable KC-SR dark state.\\

Since $n_j|0\rangle^{\otimes N}=0$,
\begin{equation}
  \langle W_j\rangle_{\rho_D}
  =\langle0|^{\otimes N}(1-n_{j-1})\,n_j\,|0\rangle^{\otimes N}=0 ,
\end{equation}
contradicting $\langle W_j\rangle_{\rho_D}>0$. Hence $\rho_D$ is entangled.
\end{proof}

\section{Finite-size scaling}
\begin{figure}
\centering
\includegraphics[width=0.5\textwidth]{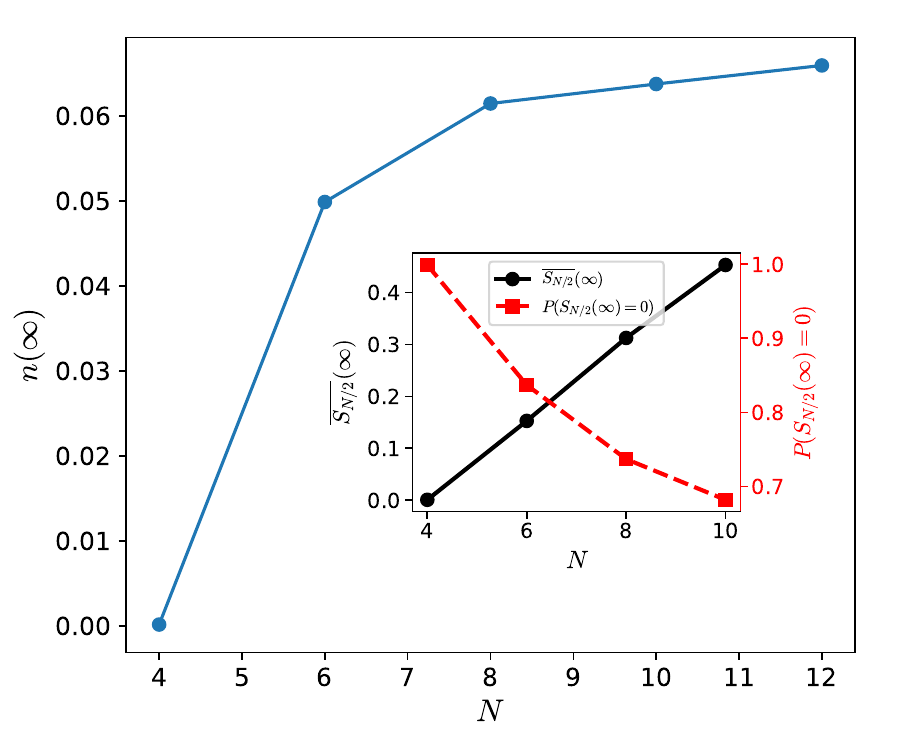}
    \caption{Finite-size analysis of excitation trapping and entanglement in the
kinetically constrained model. Results are obtained in the weak
interaction regime $(\Delta/J = 0.2)$, starting from the fully inverted
state, and averaged over $10^3$ quantum trajectories. Left panel (a): steady-state excitation density $n$ as a function of system size. Right Panel (b) Main panel: distribution of the half-chain von Neumann entropy
$S_{N/2}$ over stochastic trajectories for different system sizes.
Inset: trajectory-averaged entropy (black curve) together with the probability of trivial trajectories (red
curve).}
    \label{fig:finite_size}
\end{figure}
To assess the robustness of our results with increasing system size, we analyze finite-size scalings initializing dynamics from a fully inverted initial state, $\ket{\uparrow}^{\otimes N}$. The non-trivial steady state manifests through excitation trapping. In Fig.~\ref{fig:finite_size} {(main panel)} we show the steady-state excitation density,
\begin{equation}
n = \frac{1}{N}\sum_{j=1}^N \frac{\sigma_j^z + 1}{2},
\end{equation}
evaluated at the density-matrix level as a function of $N$. The data indicate convergence toward a finite value, supporting an extensive amount of trapped excitations.

\

We further characterize entanglement generation. {We compute} the half-chain von Neumann entropy as a proxy for pure state bi-partite entanglement generation \cite{NielsenChuang2010},
\begin{equation}
S_{N/2}(\rho^s) = - \mathrm{Tr}\!\left(\rho^s_{N/2}\log \rho^s_{N/2}\right),
\end{equation}
where $\rho^s=\ket{\psi^s(t)}\bra{\psi^s(t)}$ and $\rho^s_{N/2}=\mathrm{Tr}_{N/2}(\rho^s)$, evaluated over an ensemble of $M$ stochastic realizations for different $N$. With increasing system size, the secondary peak broadens and shifts to larger entropy values. 

\

The inset of Fig.~\ref{fig:finite_size} shows the trajectory-averaged entropy (black curve) together with the probability of trivial trajectories (red curve) which is related to the probability of finding a trajectory with $n = 0$. The finite-size analysis supports this conclusion and further reveals that larger systems sustain increasingly entangled steady state. Furthermore, these results are unchanged when interactions are stronger ($J \gtrsim \Delta$).

\end{document}